\documentclass[11pt]{article}
\usepackage{amscd,amsmath,amstext,amsfonts,amsbsy,amssymb,amsthm}
\usepackage{multicol,graphicx,framed,wasysym}
\usepackage{array,multirow,longtable,fancyhdr,booktabs}
\usepackage[unicode,colorlinks,linkcolor = red,citecolor = blue]{hyperref}
\usepackage{natbib}
\usepackage[margin=2.5cm]{geometry}
\usepackage[ddmmyyyy,hhmmss]{datetime}

\numberwithin{equation}{section}

\usepackage{array}
\usepackage{arydshln}

\DeclareMathOperator*{\argmin}{arg\,min}

\newtheorem{theorem}{{\bf Theorem}}[section]
\theoremstyle{definition} 
\theoremstyle{plain}

\theoremstyle{definition}

\newcommand{\ud}{\mathrm{d}}
\newcommand{\E}{\mathbb{E}}
\newcommand{\V}{\mathbb{V}\mathrm{ar}}
\newcommand{\aV}{\mathrm{as}\mathbb{V}\mathrm{ar}}
\newcommand{\Cov}{\mathbb{C}\mathrm{ov}}
\newcommand{\aCov}{\mathrm{as}\mathbb{C}\mathrm{ov}}
\newcommand{\Pro}{\mathrm{Pr}}
\newcommand{\Se}{\mathrm{Se}}
\newcommand{\Sp}{\mathrm{Sp}}

\newcommand{\TCF}{\mathrm{TCF}}

\title{Nonparametric Estimation of ROC Surfaces Under Verification Bias}
\author{K. To Duc\thanks{toduc@stat.unipd.it}, M. Chiogna\thanks{monica.chiogna@unipd.it}\, \& G. Adimari\thanks{gianfranco.adimari@unipd.it} \\
\small Department of Statistical Sciences, University of Padua, \\ \small Via C. Battisti, 241-243, 35121 Padova, Italy}
\date{\today}

\begin{document}
\maketitle

\pagestyle{fancy}
\fancyhf{}
\rhead{\it Draft: \today\ at \currenttime}
\cfoot{\thepage}
\begin{abstract}
Verification bias is a well known problem when the predictive ability of a diagnostic test has to be evaluated. In this paper, we discuss how to assess the accuracy of continuous-scale diagnostic tests in the presence of verification bias, when a three-class disease status is considered. In particular, we propose a fully nonparametric verification bias-corrected estimator of the ROC surface. Our approach is based on nearest-neighbor imputation and adopts generic smooth regression models for both the disease and the verification processes. Consistency and asymptotic normality of the proposed estimator are proved and its finite sample behavior is investigated by means of several Monte Carlo simulation studies. Variance estimation is also discussed and an illustrative example is presented. 
\\

\end{abstract}
\section{Introduction} \label{sec:intro}
The evaluation of the accuracy of diagnostic tests is an important issue in modern medicine. In order to evaluate a test, knowledge of the true disease status of subjects or patients under study is necessary. Usually, this is obtained by a gold standard (GS) test, or reference test, that always correctly ascertains the true disease status. 

Sensitivity (Se) and specificity (Sp) are frequently used to assess the accuracy of diagnostic tests when the disease status has two categories (e.g., ``healthy'' and ``diseased''). 
In a two-class problem, for a diagnostic test $T$ that yields a continuous measure, the receiver operating characteristic (ROC) curve is a popular tool for displaying the ability of the test to distinguish between non--diseased and diseased subjects. The ROC curve is defined as the set of points $\{(1-\Sp(c),\Se(c)), c \in (-\infty,\infty)\}$ in the unit square, where $\Se(c) = \Pro(T \ge c | \text{ subject is diseased})$ and $\Sp(c) = \Pro(T < c | \text{ subject is non--diseased})$ for given a cut point $c$. The shape of ROC curve allows to evaluate the ability of the test. For example, a ROC curve equal to a straight line joining points $(0,0)$ and $(1,1)$ represents a diagnostic test which is the random guess. A commonly used summary measure that aggregates performance information of the test is the area under ROC curve (AUC). 
Reasonable values of AUC range from 0.5, suggesting that the test is no better than chance alone, to 1.0, which indicates a perfect test.

In some medical studies, however, the disease status often involves more than two categories; for example, Alzheimer's  dementia can be classified into three categories (see \cite{chi} for more details). In such situations, quantities used to evaluate the accuracy of tests are the true class fractions (TCF's). These are well defined as a generalization of sensitivity and specificity.
%
For given a pair of cut points $(c_1,c_2)$ such that $c_1 < c_2$, the true class fractions TCF's of the continuous test $T$ at $(c_1,c_2)$ are
\begin{eqnarray}
\TCF_1(c_1) &=& \Pro(T < c_1|\mathrm{class}\, 1) = 1 - \Pro(T \ge c_1|\mathrm{class}\, 1), \nonumber \\
\TCF_2(c_1,c_2) &=& \Pro(c_1 < T < c_2|\mathrm{class}\, 2) = \Pro(T \ge c_1|\mathrm{class}\, 2) - \Pro(T \ge c_2|\mathrm{class}\, 2), \nonumber \\
\TCF_3(c_2) &=& \Pro(T > c_2|\mathrm{class}\, 3) = \Pro(T \ge c_2|\mathrm{class}\, 3) \nonumber.
\end{eqnarray}
The plot of (TCF$_1$, TCF$_2$, TCF$_3$) at various values of the pair $(c_1,c_2)$ produces the ROC surface in the unit cube. It is not hard to realize that ROC surface is a generalization of the ROC curve (see \cite{scu,nak:04,nak:14}). Indeed, the projection of the ROC surface to the plane defined by TCF$_2$ versus TCF$_1$ yields the ROC curve between classes 1 and 2. Similarly, by projecting ROC surface to the plane defined by the axes TCF$_2$ and TCF$_3$, the ROC curve between classes 2 and 3 is produced. The ROC surface will be the triangular plane with vertices $(0, 0, 1), (0, 1, 0)$, and $(1, 0, 0)$ if all of three TCF's are equal for every pair $(c_1,c_2)$. In this case, we say that the diagnostic test is the random guess, again. In practice, one can imagine that the graph of the ROC surface lies in the unit cube and above the plane of the triangle with three vertices $(0, 0, 1), (0, 1, 0)$, and $(1, 0, 0)$. A summary of the overall diagnostic accuracy of the test under consideration is the volume under the ROC surface (VUS), which can be seen as a generalization of the AUC. Reasonable values of VUS vary from 1/6 to 1, ranging from bad to perfect diagnostic tests.

If we know the true disease status of all patients for which the test $T$ is measured,  then the ROC curve or the ROC surface can be estimated unbiasedly. In practice, however,  the GS test can be too expensive, or too invasive, or both for regular use. Typically, only a subset of patients undergoes disease verification, and the decision to send a patient to verification is often based on the diagnostic test result and other patient characteristics. For example, subjects with negative test results may be less likely to receive a GS test than subjects with positive test results. If only data from patients with verified disease status are used to estimate the ROC curve or the ROC surface, this generally leads to a biased evaluation of the ability of the diagnostic tests. This bias is known as verification bias. See, for example, \cite{zho} and \cite{pep} as general references.

Correcting for verification bias is a fascinating issue of medical statistics. Various methods have been developed to deal with the problem, most of which assume that the true disease status, if missing, is missing at random (MAR), see \cite{lit}. Under the MAR assumption, there are some verification bias-corrected methods for diagnostic tests, in the two-class case. Among the others, \cite{zho} present maximum likelihood approaches, \cite{rot} consider a doubly robust estimation of the area under ROC curve, while \cite{he} study a robust estimator for sensitivity and specificity by using propensity score stratification. Verification bias correction for continuous tests has been studied by \cite{alo:05} and \cite{adi:15}. In particular, \cite{alo:05} propose four types of partially parametric estimators of sensitivity and specificity under the MAR assumption, i.e., full imputation (FI), mean score imputation (MSI), inverse probability weighting (IPW) and semiparametric efficient (SPE, also known as doubly robust DR) estimator. \cite{adi:15}, instead, proposed a fully nonparametric approach for ROC analysis.

The issue of correcting for the verification bias in ROC surface analysis is very scarcely considered in the literature. Until now, only \cite{chi} and \cite{toduc:15} discuss the issue. \cite{chi} propose  maximum likelihood estimates for ROC surface and VUS corresponding to ordinal diagnostic tests, whereas \cite{toduc:15} extend the methods in \cite{alo:05} to the estimation of ROC surfaces in cases of continuous diagnostic tests. 

FI, MSI, IPW and SPE methods in
\cite{toduc:15} are partially parametric methods. Their use requires the specification of  parametric regression models for the probability of a subject being correctly classified with respect to the disease state, or the probability of a subject being verified
(i.e., tested by GS), or both. A wrong specification of such parametric models 
can negatively affect the behavior of the estimators, 
that are no longer consistent. 


In this paper, we propose a fully nonparametric approach to estimate TCF$_1$, TCF$_2$ and TCF$_3$ in the presence of verification bias, for continuous diagnostic tests.  The proposed approach is based on a nearest-neighbor (NN) imputation rule, as in \cite{adi:15}. Consistency and asymptotic normality of the estimators derived from the proposed method are studied. In addition, estimation of their variance is also discussed. To show usefulness of our proposal and advantages in comparison with partially parametric estimators, we conduct some simulation studies and give an illustrative example.
 
The rest of paper is organized as follows. In Section 2, we review partially parametric methods for correcting for verification bias in case of continuous tests. The proposed nonparametric method for estimating  ROC surfaces and the related asymptotic results are presented in Section 3. In Section 4, we discuss variance-covariance estimation and in Section 5 we give some simulation results. An application is illustrated in Section 6. Finally, conclusions are drawn in Section 7.

\section{Partially parametric estimators of ROC surfaces}\label{sec:2}
Consider a study with $n$ subjects, for whom the result of a continuous diagnostic test $T$ is available. For  each subject,  $D$ denotes the true disease status, 
that can be possibly unknown.
 Hereafter, we will describe the true disease status as a trinomial random vector $D = (D_1,D_2,D_3)$. $D_k$ is a binary variable that takes $1$ if the subject belongs to class $k$, $k = 1,2,3$ and $0$ otherwise. Here, class 1, class 2 and class 3 can be referred, for example, as ``non-diseased'', ``intermediate'' and ``diseased''. Further, let $V$ be a binary verification status for a subject, such that $V = 1$ if he/she is undergoes the GS test, and $V = 0$ otherwise. In practice, some information, other than the results from the  test $T$, can be obtained for each patient. Let $A$ be the covariate vector for the patients, that may be associated  both with $D$ and $V$.
We are interested in estimating the ROC surface of $T$, and hence the
 true class factions $\TCF_1(c_1) = \Pro(T_i < c_1| D_{1i} = 1),$ $\TCF_1(c_1) = \Pro(c_1 < T_i < c_2| D_{2i} = 1)$ and $\TCF_3(c_1) = \Pro(T_i \ge c_2| D_{3i} = 1)$, for fixed constants $c_1, c_2$, with $c_1 < c_2$.

When all patients have their disease status verified by a GS, i.e., $V_i = 1$ for all $i = 1,\ldots,n$, for any pair of cut points $(c_1,c_2)$, the true class fractions $\TCF_1(c_1), \TCF_2(c_1,c_2)$ and $\TCF_3(c_2)$ can be easily estimated by\\
\begin{eqnarray}
\widehat{\TCF}_1(c_1) &=& 1 - \frac{\sum\limits_{i=1}^{n}\mathrm{I}(T_i \ge c_1)D_{1i}}{\sum\limits_{i=1}^{n}D_{1i}} \nonumber\\
\widehat{\TCF}_2(c_1,c_2) &=& \frac{\sum\limits_{i=1}^{n}\mathrm{I}(c_1 \le T_i < c_2)D_{2i}}{\sum\limits_{i=1}^{n}D_{2i}}\nonumber \label{nonp:est}\\
\widehat{\TCF}_3(c_2) &=& \frac{\sum\limits_{i=1}^{n}\mathrm{I}(T_i \ge c_2)D_{3i}}{\sum\limits_{i=1}^{n}D_{3i}} \nonumber,
\end{eqnarray}
where $\mathrm{I}(\cdot)$ is the indicator function. It is straightforward to show that the above estimators are unbiased. However, they cannot be employed in case of incomplete data, i.e. when $V_i = 0$ for some $i = 1,\ldots,n$.

When only some subjects are selected to undergo the GS test, we need to make an assumption about the selection mechanism. We assume that the verification status $V$ and the disease status $D$ are mutually independent given the test result $T$ and covariate $A$. This means that  $\Pro (V|T,A) = \Pro (V|D,T,A)$ or equivalently $\Pro (D|T,A) = \Pro (D|V,T,A)$. Such assumption is a special case of the missing at random (MAR) assumption (\cite{lit}).

Under MAR assumption, verification bias-corrected estimation of the  true class factions is discussed in \cite{toduc:15}, where
 (partially) parametric estimators, based on four different approaches, are given. In particular,
full imputation (FI) estimators of $\TCF_1(c_1), \TCF_2(c_1,c_2)$ and $\TCF_3(c_2)$ are defined as
\begin{eqnarray}
\widehat{\TCF}_{1,\mathrm{FI}}(c_1) &=& 1 - \frac{\sum\limits_{i=1}^{n}\mathrm{I}(T_i \ge c_1)\hat{\rho}_{1i}}{\sum\limits_{i=1}^{n}\hat{\rho}_{1i}}  \nonumber, \\
\widehat{\TCF}_{2,\mathrm{FI}}(c_1,c_2) &=& \frac{\sum\limits_{i=1}^{n}\mathrm{I}(c_1 \le T_i < c_2)\hat{\rho}_{2i}}{\sum\limits_{i=1}^{n}\hat{\rho}_{2i}}, \label{est:fi} \\
\widehat{\TCF}_{3,\mathrm{FI}}(c_2) &=& \frac{\sum\limits_{i=1}^{n}\mathrm{I}(T_i \ge c_2)\hat{\rho}_{3i}}{\sum\limits_{i=1}^{n}\hat{\rho}_{3i}}. \nonumber
\end{eqnarray}
This method requires a parametric model (e.g. multinomial logistic regression model) to obtain the estimates $\hat{\rho}_{ki}$ of $\rho_{ki} = \Pro(D_{ki} = 1|T_i, A_i)$, using only data from verified subjects. Differently, the mean score imputation (MSI) approach only uses the estimates $\hat{\rho}_{ki}$ for the missing values of disease status $D_{ki}$. Hence, MSI estimators are
\begin{eqnarray}
\widehat{\TCF}_{1,\mathrm{MSI}}(c_1) &=& 1 - \frac{\sum\limits_{i=1}^{n}\mathrm{I}(T_i \ge c_1)\left[V_i D_{1i} + (1 - V_i) \hat{\rho}_{1i}\right]}{\sum\limits_{i=1}^{n}\left[V_i D_{1i} + (1 - V_i) \hat{\rho}_{1i}\right]}, \nonumber \\
\widehat{\TCF}_{2,\mathrm{MSI}}(c_1,c_2) &=&  \frac{\sum\limits_{i=1}^{n}\mathrm{I}(c_1 \le T_i < c_2)\left[V_i D_{2i} + (1 - V_i) \hat{\rho}_{2i}\right]}{\sum\limits_{i=1}^{n}\left[V_i D_{2i} + (1 - V_i) \hat{\rho}_{2i}\right]}, \label{est:msi3} \\
\widehat{\TCF}_{3,\mathrm{MSI}}(c_2) &=&  \frac{\sum\limits_{i=1}^{n}\mathrm{I}(T_i \ge c_2)\left[V_i D_{3i} + (1 - V_i) \hat{\rho}_{3i}\right]}{\sum\limits_{i=1}^{n}\left[V_i D_{3i} + (1 - V_i) \hat{\rho}_{3i}\right]}. \nonumber
\end{eqnarray}
The inverse probability weighting (IPW) approach weights each verified subject by the inverse of the probability that the subject is selected for verification. Thus, $\TCF_1(c_1), \TCF_2(c_1,c_2)$ and $\TCF_3(c_2)$ are estimated by
\begin{eqnarray}
\widehat{\TCF}_{1,\mathrm{IPW}}(c_1) &=&  1 - \frac{\sum\limits_{i=1}^{n}\mathrm{I}(T_i \ge c_1) V_i \hat{\pi}_i^{-1}D_{1i}}{\sum\limits_{i=1}^{n}V_i \hat{\pi}_i^{-1}D_{1i}} \nonumber, \\
\widehat{\TCF}_{2,\mathrm{IPW}}(c_1,c_2) &=&  \frac{\sum\limits_{i=1}^{n}\mathrm{I}(c_1 \le T_i < c_2)V_i \hat{\pi}_i^{-1}D_{2i}}{\sum\limits_{i=1}^{n}V_i \hat{\pi}_i^{-1}D_{2i}}, \label{est:ipw4} \\
\widehat{\TCF}_{3,\mathrm{IPW}}(c_2) &=& \frac{\sum\limits_{i=1}^{n}\mathrm{I}(T_i \ge c_2) V_i \hat{\pi}_i^{-1}D_{3i}}{\sum\limits_{i=1}^{n}V_i \hat{\pi}_i^{-1}D_{3i}}, \nonumber
\end{eqnarray}
where $\hat{\pi}_i$ is an estimate of the conditional verification probabilities $\pi_i = \Pro(V_i = 1|T_i, A_i)$. Finally, the semiparametric efficient (SPE) estimators are
\begin{eqnarray}
\widehat{\TCF}_{1,\mathrm{SPE}}(c_1) &=& 1 - \frac{\sum\limits_{i=1}^{n} \mathrm{I}(T_i \ge c_1) \left\{\frac{V_i D_{1i}}{\hat{\pi}_i} - \frac{\hat{\rho}_{1i}(V_i - \hat{\pi}_i)}{\hat{\pi}_i} \right\}} {\sum\limits_{i=1}^{n} \left\{\frac{V_i D_{1i}}{\hat{\pi}_i} - \frac{\hat{\rho}_{1i}(V_i - \hat{\pi}_i)}{\hat{\pi}_i} \right\}} \nonumber, \\
\widehat{\TCF}_{2,\mathrm{SPE}}(c_1,c_2) &=&  \frac{\sum\limits_{i=1}^{n} \mathrm{I}(c_1 \le T_i < c_2) \left\{\frac{V_i D_{2i}}{\hat{\pi}_i} - \frac{\hat{\rho}_{2i}(V_i - \hat{\pi}_i)}{\hat{\pi}_i} \right\}}{\sum\limits_{i=1}^{n} \left\{\frac{V_i D_{2i}}{\hat{\pi}_i} - \frac{\hat{\rho}_{2i}(V_i - \hat{\pi}_i)}{\hat{\pi}_i} \right\}}, \label{est:spe3} \\
\widehat{\TCF}_{3,\mathrm{SPE}}(c_2) &=&  \frac{\sum\limits_{i=1}^{n} \mathrm{I}(T_i \ge c_2) \left\{\frac{V_i D_{3i}}{\hat{\pi}_i} - \frac{\hat{\rho}_{3i}(V_i - \hat{\pi}_i)}{\hat{\pi}_i} \right\}} {\sum\limits_{i=1}^{n} \left\{\frac{V_i D_{3i}}{\hat{\pi}_i} - \frac{\hat{\rho}_{3i}(V_i - \hat{\pi}_i)}{\hat{\pi}_i} \right\}}. \nonumber
\end{eqnarray}
Estimators (\ref{est:fi})-(\ref{est:spe3}) represent an extension to the three-classes problem of the estimators proposed in \cite{alo:05}. SPE estimators are also known to be doubly robust estimators, in the sense that they are consistent if either the $\rho_{ki}$'s or the $\pi_{i}$'s are estimated consistently. However, SPE estimates could fall outside the interval $(0,1)$. This happens because the quantities $V_i D_{ki}\hat{\pi}_i^{-1} - \hat{\rho}_{ki}(V_i - \hat{\pi}_i)\hat{\pi}_i^{-1}$ can be negative.

\section{Nonparametric estimators}\label{s:prop}
\subsection{The proposed method}
All the verification bias-corrected estimators of $\TCF_1(c_1), \TCF_2(c_1,c_2)$ and $\TCF_3(c_2)$ revised in the previous section belong to the class of (partially) parametric estimators, i.e., they need regression models to estimate $\rho_{ki} = \Pro(D_{ki} = 1|T_i, A_i)$ and/or $\pi_i = \Pro(V_i = 1|T_i, A_i)$. 
In what follows, we propose a fully nonparametric approach to the estimation of $\TCF_1(c_1), \TCF_2(c_1,c_2)$ and $\TCF_3(c_2)$. Our approach is based on the K-nearest neighbor (KNN) imputation method. Hereafter, we shall assume that $A$ is a continuous random variable.

Recall that the true disease status is a trinomial random vector $D = (D_1,D_2,D_3)$ such that $D_k$ is a $n$ Bernoulli trials with success probability $\theta_k = \Pro(D_k = 1)$. Note that $\theta_1 + \theta_2 + \theta_3 = 1$. Let $\beta_{jk} = \Pro(T \ge c_j, D_k = 1)$ with $j = 1,2$ and $k = 1,2,3$.
Since parameters $\theta_k$ are the means of the random variables $D_{k}$,  we can use the KNN estimation procedure discussed in (\cite{ning:12}) to obtain nonparametric estimates $\hat{\theta}_{k,\mathrm{KNN}}$. More precisely, we define
\begin{equation}
\hat{\theta}_{k,\mathrm{KNN}} = \frac{1}{n}\sum_{i=1}^{n}\left[V_i D_{ki} + (1-V_i)\hat{\rho}_{ki,K}\right], \qquad K \in \mathbb{N},
\nonumber \label{est:knn1}
\end{equation}
where $\hat{\rho}_{ki,K} = \dfrac{1}{K}\sum\limits_{l=1}^{K}D_{ki(l)}$, and $\left\{(T_{i(l)},A_{i(l)},D_{ki(l)}): V_{i(l)} = 1, l = 1,\ldots,K\right\}$ is a set of $K$ observed data pairs and $(T_{i(l)},A_{i(l)})$ denotes the $j$-th nearest neighbor to $(T_i,A_i)$ among all $(T,A)$'s corresponding to the verified patients, i.e., to those $D_{kh}$'s with $V_h = 1$. Similarly, we can define the KNN estimates of $\beta_{jk}$ as follows
\begin{equation}
\hat{\beta}_{jk,\mathrm{KNN}} = \frac{1}{n}\sum_{i=1}^{n}\mathrm{I}(T_i \ge c_j)\left[V_iD_{ki} + (1-V_i)\hat{\rho}_{ki,K}\right],
\nonumber \label{est:knn2}
\end{equation}
each $j,k$. Therefore, the KNN imputation estimators for $\TCF_k$ are
\begin{align}
\widehat{\TCF}_{1,\mathrm{KNN}}(c_1) &= 1 - \frac{\hat{\beta}_{11}}{\hat{\theta}_1} =  \frac{\sum\limits_{i=1}^{n}\mathrm{I}(T_i < c_1)\left[V_i D_{1i} + (1 - V_i) \hat{\rho}_{1i,K}\right]}{\sum\limits_{i=1}^{n}\left[V_i D_{1i} + (1 - V_i) \hat{\rho}_{1i,K}\right]} \nonumber, \displaybreak[3] \\
\widehat{\TCF}_{2,\mathrm{KNN}}(c_1,c_2) &= \frac{\hat{\beta}_{12} - \hat{\beta}_{22}}{\hat{\theta}_2} = \frac{\sum\limits_{i=1}^{n}\mathrm{I}(c_1 \le T_i < c_2)\left[V_i D_{2i} + (1 - V_i) \hat{\rho}_{2i,K}\right]}{\sum\limits_{i=1}^{n}\left[V_i D_{2i} + (1 - V_i) \hat{\rho}_{2i,K}\right]},  \label{est:knn3}  \\
\widehat{\TCF}_{3,\mathrm{KNN}}(c_2) &= \frac{\hat{\beta}_{23}}{\hat{\theta}_3} =  \frac{\sum\limits_{i=1}^{n}\mathrm{I}(T_i \ge c_2)\left[V_i D_{3i} + (1 - V_i) \hat{\rho}_{3i,K}\right]}{\sum\limits_{i=1}^{n}\left[V_i D_{3i} + (1 - V_i) \hat{\rho}_{3i,K}\right]}. \nonumber
\end{align}

\subsection{Asymptotic distribution}
Let $\rho_k(t,a) = \Pro(D_k = 1|T=t,A=a)$ and $\pi(t,a) = \Pro(V=1|T=t,A=a)$. 
The KNN imputation estimators of $\TCF_1(c_1), \TCF_2(c_1,c_2)$ and $\TCF_3(c_2)$ are consistent and asymptotically normal. In fact, we have the following theorems.
\begin{theorem}\label{thr:knn:1}
Assume the functions $\rho_k(t,a)$ and $\pi(t,a)$ are finite and first-order differentiable. Moreover, assume that the expectation of $1/\pi(T,A)$ exists. Then, for a fixed pair cut of points $(c_1,c_2)$ such that $c_1 < c_2$, the KNN imputation estimators $\widehat{\TCF}_{1,\mathrm{KNN}}(c_1)$, $\widehat{\TCF}_{2,\mathrm{KNN}}(c_1,c_2)$ and $\widehat{\TCF}_{3,\mathrm{KNN}}(c_2)$ are consistent.
\end{theorem}
\begin{proof}
Since the disease status $D_{k}$ is a Bernoulli random variable, its second-order moment, $\E (D_{k}^2)$, is finite. According to the first assumption, we can show that the conditional variance of $D_{k}$ given the test results $T$ and $A$, $\V (D_k | T = t, A = a)$ is equal to $\rho_{k}(t,a)\left[1 - \rho_{k}(t,a)\right]$ and is clearly finite. Thus, by an application of Theorem 1 in \cite{ning:12}, the KNN imputation estimators $\hat{\theta}_{k,\mathrm{KNN}}$ are consistent. 

Now, observe that,
\begin{eqnarray}
\hat{\beta}_{jk,\mathrm{KNN}} - \beta_{jk} &=& \frac{1}{n}\sum_{i=1}^{n}\mathrm{I}(T_i \ge c_j)\left[V_iD_{ki} + (1-V_i)\rho_{ki}\right] + \frac{1}{n}\sum_{i=1}^{n}\mathrm{I}(T_i \ge c_j) (1-V_i)(\hat{\rho}_{ki,K} - \rho_{ki}) - \beta_{jk} \nonumber\\
&=& \frac{1}{n}\sum_{i=1}^{n}\mathrm{I}(T_i \ge c_j)V_i\left[D_{ki} - \rho_{ki}\right] + \frac{1}{n}\sum_{i=1}^{n}\left[\mathrm{I}(T_i \ge c_j)\rho_{ki} - \beta_{jk}\right]\nonumber \\
&& + \: \frac{1}{n}\sum_{i=1}^{n}\mathrm{I}(T_i \ge c_j) (1-V_i)(\hat{\rho}_{ki,K} - \rho_{ki})\nonumber \\
&=& S_{jk} + R_{jk} + T_{jk}. \nonumber
\end{eqnarray}
Here, the quantities $R_{jk}, S_{jk}$ and $T_{jk}$ are similar to the quantities $R,S$ and $T$ in the proof of Theorem 2.1 in \cite{cheng:94} and Theorem 1 in \cite{ning:12}. Thus, we have that
\[
\sqrt{n}R_{jk} \stackrel{d}{\to}\mathcal{N}\left(0,\V\left[\mathrm{I}(T \ge c_j)\rho_k(T,A)\right]\right) \qquad \text{and} \qquad \sqrt{n}S_{jk} \stackrel{d}{\to} \mathcal{N}\left(0,\E \left[\pi(T,A)\delta^2_{jk}(T,A)\right]\right),
\]
where $\delta_{jk}^2(T,A)$ is the conditional variance of $\mathrm{I}(T \ge c_j,D_{k} = 1)$ given $T,A$. Also, by using a similar technique to that of proof of Theorem 1 in \cite{ning:12}, we get $T_{jk} = W_{jk} + o_{p}(n^{-1/2})$, where
\[
W_{jk} = \frac{1}{n}\sum_{i=1}^{n}\mathrm{I}(T_i \ge c_j)(1-V_i)\left[\frac{1}{K}\sum_{l=1}^{K}\left(V_{i(l)}D_{ki(l)} - \rho_{ki(l)}\right)\right].
\]
Moreover, $\E(W_{jk}) = 0$ and
\[
\aV(\sqrt{n}W_{jk}) = \frac{1}{K}\E \left[(1 - \pi(T,A))\delta_{jk}^2(T,A)\right] + \E \left[\frac{(1-\pi(T,A))^2\delta_{jk}^2(T,A)}{\pi(T,A)}\right].
\]
Then, the Markov's inequality implies that $W_{jk} \stackrel{p}{\to} 0$ as $n$ goes to infinity. This, together with the fact that $R_{jk}$ and $S_{jk}$ converge in probability to zero, leads to the consistency of $\hat{\beta}_{jk,\mathrm{KNN}}$, i.e, $\hat{\beta}_{jk,\mathrm{KNN}} \stackrel{p}{\to} \beta_{jk}$. It follows that  $\widehat{\TCF}_{1,\mathrm{KNN}}(c_1) = 1 - \frac{\hat{\beta}_{11}}{\hat{\theta}_1}$, $\widehat{\TCF}_{2,\mathrm{KNN}}(c_1,c_2) = \frac{\hat{\beta}_{12} - \hat{\beta}_{22}}{\hat{\theta}_2}$ and $\widehat{\TCF}_{3,\mathrm{KNN}}(c_2) = \frac{\hat{\beta}_{23}}{\hat{\theta}_3}$ are consistent.
\end{proof}
\begin{theorem}\label{thr:knn:2}
Assume that the conditions in Theorem \ref{thr:knn:1} hold, we get
\begin{equation}
\sqrt{n}\left[\begin{pmatrix}
\widehat{\TCF}_{1,\mathrm{KNN}}(c_1) \\ \widehat{\TCF}_{2,\mathrm{KNN}}(c_1,c_2) \\
\widehat{\TCF}_{3,\mathrm{KNN}}(c_2)
\end{pmatrix} - \begin{pmatrix}
\TCF_{1}(c_1) \\ \TCF_{2}(c_1,c_2) \\
\TCF_{3}(c_2)
\end{pmatrix}\right] \stackrel{d}{\to} \mathcal{N}(0,\Xi),
\label{asy:knn}
\end{equation}
where $\Xi$ is a suitable matrix.
\end{theorem}
\begin{proof}
A direct application of Theorem 1 in \cite{ning:12} gives the result that the quantity $\sqrt{n}(\hat{\theta}_{k,\mathrm{KNN}} - \theta_k)$ converges to a normal random variable with mean $0$ and variance $\sigma^2_{k} = \left[\theta_k(1-\theta_k) + \omega_{k}^2\right]$. Here,
\begin{eqnarray}
\omega_{k}^2 &=& \left(1 + \frac{1}{K}\right)\E \left[\rho_k(T,A)(1-\rho_k(T,A))(1-\pi(T,A))\right] \nonumber\\
&& + \: \E \left[\frac{\rho_k(T,A)(1-\rho_k(T,A))(1-\pi(T,A))^2}{\pi(T,A)}\right]. \label{omega_k}
\end{eqnarray}
In addition, from the proof of Theorem \ref{thr:knn:1}, we have
\[
\hat{\beta}_{jk,\mathrm{KNN}} - \beta_{jk} \simeq S_{jk} + R_{jk} + W_{jk} + o_{p}(n^{-1/2}),
\]
with
\[
\sqrt{n}R_{jk} \stackrel{d}{\to}\mathcal{N}\left(0,\V\left[\mathrm{I}(T \ge c_j)\rho_k(T,A)\right]\right), \quad \sqrt{n}S_{jk} \stackrel{d}{\to} \mathcal{N}\left(0,\E \left[\pi(T,A)\delta^2_{jk}(T,A)\right]\right)
\]
and
\[
\sqrt{n}W_{jk} \stackrel{d}{\to} \mathcal{N}(0,\sigma^2_{W_{jk}}).
\]
Therefore, $\sqrt{n}(\hat{\beta}_{jk,\mathrm{KNN}} - \beta_{jk}) \stackrel{d}{\to} \mathcal{N}(0,\sigma_{jk}^2)$. Here, the asymptotic variance $\sigma^2_{jk}$ is obtained by
\[
\sigma_{jk}^2 = 
\left[\beta_{jk}\left(1 - \beta_{jk}\right) + \omega_{jk}^2\right],
\]
with
\begin{eqnarray}
\omega_{jk}^2 &=& \left(1 + \frac{1}{K}\right)\E \left[\mathrm{I}(T \ge c_j)\rho_k(T,A)(1-\rho_k(T,A))(1-\pi(T,A))\right] \nonumber\\
&& + \: \E \left[\frac{\mathrm{I}(T \ge c_j)\rho_k(T,A)(1-\rho_k(T,A))(1-\pi(T,A))^2}{\pi(T,A)}\right]. \label{omega_jk}
\end{eqnarray}
This result follows by the fact that $R_{jk}$ and $S_{jk} + W_{jk}$ are uncorrelated and the asymptotic covariance between $S_{jk}$ and $W_{jk}$ is obtained by
\[
\aCov\left(S_{jk},W_{jk}\right) = \E \left[(1 - \pi(T,A))\delta_{jk}^2(T,A)\right].
\]
Moreover, we get that the vector $\sqrt{n}(\hat{\theta}_{1,\mathrm{KNN}},\hat{\theta}_{2,\mathrm{KNN}},\hat{\beta}_{11,\mathrm{KNN}},\hat{\beta}_{12,\mathrm{KNN}},\hat{\beta}_{22,\mathrm{KNN}},\hat{\beta}_{23,\mathrm{KNN}})^\top$ is (jointly) asymptotically normally distributed with mean vector $(\theta_1,\theta_2,\beta_{11},\beta_{12},\beta_{22},\beta_{23})^\top$ and suitable covariance matrix $\Xi^*$. Then,  result \eqref{asy:knn} follows by applying the multivariate delta method to
\begin{equation}
h(\hat{\theta}_{1},\hat{\theta}_{2},\hat{\beta}_{11},\hat{\beta}_{12},\hat{\beta}_{22},\hat{\beta}_{23})= \left(1 - \frac{\hat\beta_{11}}{\hat\theta_{1}}, \frac{(\hat\beta_{12} - \hat\beta_{22})}{\hat\theta_{2}},\frac{\hat\beta_{23}}{(1 - \hat\theta_{1} - \hat\theta_{2})}\right).
\nonumber \label{asy_h:knn}
\end{equation}
 The asymptotic covariance matrix of $\sqrt{n}(\widehat{\TCF}_{1,\mathrm{KNN}},\widehat{\TCF}_{2,\mathrm{KNN}},\widehat{\TCF}_{3,\mathrm{KNN}})^\top$, $\Xi$, is obtained by
\begin{equation}
\Xi = h'\Xi^*h'^\top,
\label{asym_var:knn1}
\end{equation}
where $h'$ is the first-order derivative of $h$, i.e.,
\begin{equation}
h' = \begin{pmatrix}
\frac{\beta_{11}}{\theta_{1}^2} & 0 & -\frac{1}{\theta_{1}} & 0 & 0 & 0 \\
0 & -\frac{(\beta_{12} - \beta_{22})}{\theta_{2}^2} & 0 & \frac{1}{\theta_{2}} & -\frac{1}{\theta_{2}} & 0\\
\frac{\beta_{23}}{(1 - \theta_{1} - \theta_{2})^2} & \frac{\beta_{23}}{(1 - \theta_{1} - \theta_{2})^2} & 0 & 0 & 0 & \frac{1}{(1 - \theta_{1} - \theta_{2})}
\end{pmatrix}.
\label{asym_var:knn2}
\end{equation}
\end{proof}

\subsection{The asymptotic covariance matrix}
\label{asy_var:knn}
Let
\[
\Xi = 
\begin{pmatrix}
\xi_1^2 & \xi_{12} & \xi_{13} \\
\xi_{12} & \xi_2^2 & \xi_{23} \\
\xi_{13} & \xi_{23} & \xi^2_3
\end{pmatrix}
.
\]
The asymptotic covariance matrix $\Xi^*$ is a $6 \times 6$ matrix such that its diagonal elements are the asymptotic variances of $\sqrt{n}\hat{\theta}_{k,\mathrm{KNN}}$ and $\sqrt{n}\hat{\beta}_{jk,\mathrm{KNN}}$.
Let us define $\sigma_{12}^* = \aCov(\sqrt{n}\hat{\theta}_{1,\mathrm{KNN}},\sqrt{n}\hat{\theta}_{2,\mathrm{KNN}})$, $\sigma_{sjk} = \aCov(\sqrt{n}\hat{\theta}_{s,\mathrm{KNN}},\sqrt{n}\hat{\beta}_{jk,\mathrm{KNN}})$ and $\sigma_{jkls} = \aCov(\sqrt{n}\hat{\beta}_{jk,\mathrm{KNN}},\sqrt{n}\hat{\beta}_{ls,\mathrm{KNN}})$. We write
\begin{equation}
\Xi^* = \begin{pmatrix}
\sigma_1^2 & \sigma_{12}^* & \sigma_{111} & \sigma_{112} & \sigma_{122} & \sigma_{123} \\
\sigma_{12}^* & \sigma_2^2 & \sigma_{211} & \sigma_{212} & \sigma_{222} & \sigma_{223} \\
\sigma_{111} & \sigma_{211} & \sigma_{11}^2 & \sigma_{1112} & \sigma_{1122} & \sigma_{1123} \\
\sigma_{112} & \sigma_{212} & \sigma_{1112} & \sigma_{12}^2 & \sigma_{1222} & \sigma_{1223} \\
\sigma_{122} & \sigma_{222} & \sigma_{1122} & \sigma_{1222} & \sigma_{22}^2 & \sigma_{2223} \\
\sigma_{123} & \sigma_{223} & \sigma_{1123} & \sigma_{1223} & \sigma_{2223} & \sigma_{23}^2
\end{pmatrix}.
\nonumber \label{asym_var:knn3}
\end{equation}
Hence, from (\ref{asym_var:knn1}) and (\ref{asym_var:knn2}),
\begin{eqnarray}
\xi_1^2=\aV\left(\sqrt{n}\widehat{\TCF}_{1,\mathrm{KNN}}(c_1)\right) &=& \frac{\beta_{11}^2}{\theta_1^4}\sigma_1^2 + \frac{\sigma_{11}^2}{\theta_1^2} - 2\frac{\beta_{11}}{\theta_1^3}\sigma_{111}, \nonumber \label{aV:tcf1} \\
\xi_2^2=\aV\left(\sqrt{n}\widehat{\TCF}_{2,\mathrm{KNN}}(c_1,c_2)\right) &=& \sigma_2^2\frac{(\beta_{12} - \beta_{22})^2}{\theta_2^4} + \frac{\sigma_{12}^2 + \sigma_{22}^2 - 2\sigma_{1222}}{\theta_2^2} \nonumber\\
&& - \: 2\frac{\beta_{12} - \beta_{22}}{\theta_2^3}(\sigma_{212} - \sigma_{222}) , \nonumber \label{aV:tcf2} \\
\xi_3^2=\aV\left(\sqrt{n}\widehat{\TCF}_{3,\mathrm{KNN}}(c_2)\right) &=& \frac{\beta_{23}^2}{(1 - \theta_1 - \theta_2)^4}\left(\sigma_1^2 + 2\sigma_{12}^* + \sigma_2^2\right) + \frac{\sigma_{23}^2}{(1-\theta_1-\theta_2)^2} \nonumber\\
&& + \: 2\frac{\beta_{23}}{(1-\theta_1- \theta_2)^3}\left(\sigma_{123} + \sigma_{223}\right). \label{aV:tcf3}
\end{eqnarray}
Let $\lambda^2=\aV(\sqrt{n}\hat{\beta}_{12,\mathrm{KNN}} -\sqrt{n}\hat{\beta}_{22,\mathrm{KNN}}) $. Hence, $\sigma_{12}^2 + \sigma_{22}^2 - 2\sigma_{1222}=\lambda^2$, and
\begin{equation}
\xi_2^2=
 \sigma_2^2\frac{(\beta_{12} - \beta_{22})^2}{\theta_2^4} + \frac{\lambda^2}{\theta_2^2} - 2\frac{\beta_{12} - \beta_{22}}{\theta_2^3}(\sigma_{212} - \sigma_{222}) \nonumber \label{aV:tcf2_2}.
\end{equation}
Observe that $\hat{\theta}_{3,\mathrm{KNN}} = 1 - (\hat{\theta}_{1,\mathrm{KNN}} + \hat{\theta}_{2,\mathrm{KNN}})$. Thus,
\begin{eqnarray}
\aV(\sqrt{n}\hat{\theta}_{3,\mathrm{KNN}}) &=& \aV(\sqrt{n}\hat{\theta}_{1,\mathrm{KNN}} + \sqrt{n}\hat{\theta}_{2,\mathrm{KNN}}) \nonumber\\
&=& \aV(\sqrt{n}\hat{\theta}_{1,\mathrm{KNN}}) + \aV(\sqrt{n}\hat{\theta}_{2,\mathrm{KNN}}) + 2\aCov(\sqrt{n}\hat{\theta}_{1,\mathrm{KNN}},\sqrt{n}\hat{\theta}_{2,\mathrm{KNN}}).
\nonumber
\end{eqnarray}
This leads to the expression $\sigma_3^2=\sigma_1^2 + 2\sigma^*_{12} + \sigma_2^2$. In addition, 
\begin{eqnarray}
\sigma_{123} + \sigma_{223} &=& \aCov(\sqrt{n}\hat{\theta}_{1,\mathrm{KNN}},\sqrt{n}\hat{\beta}_{23,\mathrm{KNN}}) + \aCov(\sqrt{n}\hat{\theta}_{2,\mathrm{KNN}},\sqrt{n}\hat{\beta}_{23,\mathrm{KNN}}) \nonumber\\
&=& \aCov(\sqrt{n}\hat{\theta}_{1,\mathrm{KNN}} + \sqrt{n}\hat{\theta}_{2,\mathrm{KNN}},\sqrt{n}\hat{\beta}_{23,\mathrm{KNN}}) \nonumber \\
&=& - \aCov(\sqrt{n} - (\sqrt{n}\hat{\theta}_{1,\mathrm{KNN}} + \sqrt{n}\hat{\theta}_{2,\mathrm{KNN}}),\sqrt{n}\hat{\beta}_{23,\mathrm{KNN}}) \nonumber \\
&=& - \sigma_{323}.\nonumber
\end{eqnarray}
Therefore, from (\ref{aV:tcf3}), the asymptotic variance of $\sqrt{n}\widehat{\TCF}_{3,\mathrm{KNN}}(c_2)$ is 
\begin{equation}
\xi_3^2=
 \frac{\beta_{23}^2 \sigma_3^2}{(1 - \theta_1 - \theta_2)^4} + \frac{\sigma_{23}^2}{(1-\theta_1-\theta_2)^2} - 2\frac{\beta_{23} \sigma_{323}}{(1-\theta_1- \theta_2)^3}. \nonumber \label{aV:tcf3_2}
\end{equation}
Recall that 
$\sigma_k^2 = \left[\theta_k(1 - \theta_k) + \omega_k^2\right]$ and 
$\sigma_{jk}^2=\left[\beta_{jk}(1 - \beta_{jk}) + \omega_{jk}^2\right]$,
where $\omega_{k}^2$  and $\omega_{jk}^2$ are given in (\ref{omega_k}) and
(\ref{omega_jk}), respectively.
To obtain $\sigma_{kjk}$, we observe that
\begin{eqnarray}
\beta_{jk} &=& \Pro\left(T \ge c_j, D_k = 1\right) = \Pro\left(D_k = 1\right)\Pro\left(T \ge c_j| D_k = 1\right) \nonumber\\
&=& \Pro\left(D_k = 1\right)\left[1 - \Pro\left(T < c_j| D_k = 1\right)\right] \nonumber\\
&=& \Pro\left(D_k = 1\right) - \Pro\left(D_k = 1\right)\Pro\left(T < c_j| D_k = 1\right) \nonumber \\
&=& \Pro\left(D_k = 1\right) - \Pro\left(T < c_j, D_k = 1\right) \nonumber\\
&=& \theta_k - \gamma_{jk}\nonumber,
\end{eqnarray}
for $k = 1,2,3$ and $j = 1,2$. Then, we define
\[
\hat{\gamma}_{jk,\mathrm{KNN}} = \frac{1}{n}\sum_{i=1}^{n}\mathrm{I}(T_i < c_j)\left[V_iD'_{ki} + (1-V_i)\hat{\rho}_{ki,K}\right].
\]
The asymptotic variance of $\sqrt{n}\hat{\gamma}_{jk,\mathrm{KNN}}$, $\zeta_{jk}^2$, is obtained as that of $\sqrt{n}\hat{\beta}_{jk,\mathrm{KNN}}$. In fact, we get $\zeta_{jk}^2 = \left[\gamma_{jk}(1 - \gamma_{jk}) + \eta_{jk}^2\right]$, where
\begin{eqnarray}
\eta_{jk}^2 &=& \frac{K + 1}{K}\E \left[\mathrm{I}(T < c_j)\rho_k(T,A)\{1-\rho_k(T,A)\} \{1-\pi(T,A)\}\right] \nonumber\\
&& + \: \E \left[\frac{\mathrm{I}(T < c_j)\rho_k(T,A)\{1-\rho_k(T,A)\} \{1-\pi(T,A)\}^2}{\pi(T,A)}\right]. \nonumber \label{eta_jk}
\end{eqnarray}
It is straightforward to see that $\hat{\gamma}_{jk,\mathrm{KNN}} = \hat{\theta}_{k,\mathrm{KNN}} - \hat{\beta}_{jk,\mathrm{KNN}}$. Thus, we can compute the asymptotic covariances $\sigma_{kjk}$ for $j = 1,2$ and $k = 1,2,3$, using the fact that
\begin{eqnarray}
\aV(\sqrt{n}\hat{\gamma}_{jk,\mathrm{KNN}}) &=& \aV(\sqrt{n}\hat{\theta}_{k,\mathrm{KNN}} - \sqrt{n}\hat{\beta}_{jk,\mathrm{KNN}}) \nonumber \\
&=& \aV(\sqrt{n}\hat{\theta}_{k,\mathrm{KNN}}) + \aV(\sqrt{n}\hat{\beta}_{jk,\mathrm{KNN}}) - 2\aCov (\sqrt{n}\hat{\theta}_{k,\mathrm{KNN}}, \sqrt{n}\hat{\beta}_{jk,\mathrm{KNN}}).\nonumber
\end{eqnarray}
This leads to
\[
\sigma_{kjk} = \frac{1}{2}\left(\sigma^2_k + \sigma^2_{jk} - \zeta^2_{jk}\right) \label{sigma_kjk}.
\]
Hence,
\[
\begin{array}{r r r r}
\sigma_{111} &= \dfrac{1}{2}\left(\sigma_1^2 + \sigma_{11}^2 - \zeta_{11}^2\right) ;& \qquad \sigma_{212} &= \dfrac{1}{2}\left(\sigma_2^2 + \sigma_{12}^2 - \zeta_{12}^2\right); \\ [8pt]
\sigma_{222} &= \dfrac{1}{2}\left(\sigma_2^2 + \sigma_{22}^2 - \zeta_{22}^2\right) ;& \qquad \sigma_{323} &= \dfrac{1}{2}\left(\sigma_3^2 + \sigma_{23}^2 - \zeta_{23}^2\right).
\end{array}
\]
As for $\lambda^2$, one can show that 
\begin{equation}
\lambda^2 = \frac{1}{n}\left\{(\beta_{12} - \beta_{22})\left[1 - (\beta_{12} - \beta_{22})\right] + \omega^2_{12} - \omega^2_{22}\right\}. \nonumber 
\end{equation}
(see Appendix 1 and, in particular, equation (\ref{eq:rem:2})). 
Therefore, suitable explicit expressions for the asymptotic variances of KNN estimators can be found. Such expressions will depend on quantities as $\theta_k$, $\beta_{jk}$ $\omega^2_k$, $\omega^2_{jk}$, $\gamma_{jk}$ and $\eta^2_{jk}$ only. As a consequence, to obtain consistent estimates of the asymptotic variances, ultimately we need to estimate the quantities $\omega_k^2, \omega_{jk}^2$ and $\eta_{jk}^2$. 

In Appendix 2 we show that suitable expressions can be obtained also for the elements
$\xi_{12}$, $\xi_{13}$ and $\xi_{23}$ of the covariance matrix $\Xi$. Such expressions will depend, among others, on certain quantities $\psi^2_{1212}$, $\psi^2_{112}$, $\psi^2_{213}$, $\psi^2_{12}$, $\psi^2_{113}$, $\psi^2_{223}$ and $\psi^2_{1223}$ similar to $\omega^2_k$, $\omega^2_{jk}$ or $\eta_{jk}^2$.

\subsection{Choice of $K$ and the distance measure}
\label{sec:choice}
The proposed method is based on nearest-neighbor imputation,
which requires the choice of a value for $K$ as well as
a distance measure.

In practice,
the selection of a suitable distance is  tipically dictated
by features of the data and possible subjective evaluations; thus, a general indication about
an adequate choice is difficult to express. 
In many cases, the simple Euclidean distance may be appropriate.
Other times, the researcher may wish to consider specific characteristics of data at hand, and then make a different choice. For example, the diagnostic test result $T$ and the auxiliary covariate $A$ could be heterogeneous with respect to their variances (which is particularly true when the variables are measured on heterogeneous scales). In this case, the choice of the Mahalanobis distance may be suitable.


As for the choice of the size of the neighborhood, \citet{ning:12} argue that
 nearest-neighbor imputation whit a small value of $K$ tipically yields negligible bias of the estimators, but a large variance; the opposite happen with a large value of $K$. The authors suggest that the choice of $K \in \{1,2\}$ is generally adequate when the aim is to estimate an average. 
 A similar comment is also raised by \citet{adi:15} and \citet{adi:15:auc}, i.e., a small value of $K$, within the range 1--3, may be a good choice to estimate ROC curves and AUC. However, the authors stress that, in general, the choice of $K$ may depend on the dimension of the feature space, and propose to use cross--validation to find $K$ in case of high--dimensional covariate.
 Specifically, the authors indicate that a suitable value of the size of neighbor could be found by
\[
K^* = \argmin_{K = 1, \ldots, n_{ver}} \frac{1}{n_{ver}}\left\|D - \hat{\rho}_{K} \right\|_1,
\]
where $\|\cdot\|_1$ denotes $L_1$ norm for vector and $n_{ver}$ is the number of verified subjects. The formula above can be generalized to our multi--class case. In fact, when the disease status $D$ has $q$ categories ($q \ge 3$), the difference between $D$ and $\hat{\rho}_{K}$ is a $n_{ver} \times (q-1)$ matrix. In such situation, the selection rule could be
\begin{equation}
K^* = \argmin_{K = 1, \ldots, n_{ver}} \frac{1}{n_{ver}(q-1)}\left\|D - \hat{\rho}_{K} \right\|_{1,1},
\label{choice:K:1}
\end{equation}
where $\|{\cal A}\|_{1,1}$ denotes $L_{1,1}$ norm of matrix ${\cal A}$, i.e.,
\[
\|{\cal A}\|_{1,1} =  \sum_{j=1}^{q-1}\left(\sum_{i=1}^{n_{ver}}|a_{ij}|\right).
\]



\section{Variance-covariance estimation}

Consider first the problem of estimating of the variances of $\widehat{\TCF}_{1,\mathrm{KNN}}$, $\widehat{\TCF}_{2,\mathrm{KNN}}$ and $\widehat{\TCF}_{3,\mathrm{KNN}}$. In a nonparametric framework, quantities as ${\omega}_k^2, {\omega}_{jk}^2$ and ${\eta}_{jk}^2$  can be estimated by their empirical counterparts, using also the plug--in method. Here, we consider an approach that uses a nearest-neighbor rule to estimate both the functions $\rho_k(T,A)$ and the propensity score $\pi(T,A)$, that are present in the expressions of ${\omega}_k^2, {\omega}_{jk}^2$ and ${\eta}_{jk}^2$. In particular, for the conditional probabilities of disease, we can use KNN estimates $\tilde{\rho}_{ki}=\hat{\rho}_{ki,\bar{K}}$, where the integer $\bar{K}$ must be greater than one to avoid estimates equal to zero. For the conditional probabilities of verification, we can resort to the KNN procedure proposed in \cite{adi:15}, which considers the estimates
\[
\tilde{\pi}_{i} = \frac{1}{K^*_i}\sum_{l = 1}^{K^*_i} V_{i(l)},
\]
where $\left\{(T_{i(l)},A_{i(l)},V_{i(l)}): l = 1,\ldots,K_i^*\right\}$ is a set of $K_i^*$ observed pairs and $(T_{i(l)},A_{i(l)})$ denotes the $j$-th nearest neighbor to $(T_i,A_i)$ among all $(T,A)$'s. When $V_i$ equals 0, $K_i^*$ is set equal to the rank of the first verified nearest neighbor to the unit $i$, i.e., $K_i^*$ is such that $V_{i(K_i^*)} = 1$ and $V_i = V_{i(1)} = V_{i(2)} = \ldots = V_{i(K_i^* - 1)} = 0$. In case of $V_i = 1$, $K_i^*$ is such that $V_i = V_{i(1)} = V_{i(2)} = \ldots = V_{i(K_i^* - 1)} = 1$, and $V_{i(K_i^*)} = 0$, i.e., $K_i^*$ is set equal to the rank of the first non--verified nearest neighbor to the unit $i$. 
Such a procedure automatically avoids zero values for the
$\tilde{\pi}_{i}$'s.
%
%

Then, based on the $\tilde{\rho}_{ki}$'s and $\tilde{\pi}_{i}$'s, we obtain the estimates
\begin{eqnarray}
\hat{\omega}_{k}^2 &=& \frac{K + 1}{nK}\sum_{i=1}^{n}\tilde{\rho}_{ki}\left(1 - \tilde{\rho}_{ki}\right)\left(1 - \tilde{\pi}_{i}\right) + \frac{1}{n}\sum_{i = 1}^{n}\frac{\tilde{\rho}_{ki}\left(1 - \tilde{\rho}_{ki}\right)\left(1 - \tilde{\pi}_{i}\right)^2}{\tilde{\pi}_{i}}, \nonumber \\
\hat{\omega}_{jk}^2 &=& \frac{K + 1}{nK}\sum_{i=1}^{n}\mathrm{I}(T_i \ge c_j)\tilde{\rho}_{ki}\left(1 - \tilde{\rho}_{ki}\right)\left(1 - \tilde{\pi}_{i}\right) \nonumber \\ 
&& + \: \frac{1}{n}\sum_{i = 1}^{n}\frac{\mathrm{I}(T_i \ge c_j)\tilde{\rho}_{ki}\left(1 - \tilde{\rho}_{ki}\right)\left(1 - \tilde{\pi}_{i}\right)^2}{\tilde{\pi}_{i}}, \nonumber \\
\hat{\eta}_{jk}^2 &=& \frac{K + 1}{nK}\sum_{i=1}^{n}\mathrm{I}(T_i < c_j)\tilde{\rho}_{ki}\left(1 - \tilde{\rho}_{ki}\right)\left(1 - \tilde{\pi}_{i}\right) \nonumber \\ 
&& + \: \frac{1}{n}\sum_{i = 1}^{n}\frac{\mathrm{I}(T_i < c_j)\tilde{\rho}_{ki}\left(1 - \tilde{\rho}_{ki}\right)\left(1 - \tilde{\pi}_{i}\right)^2}{\tilde{\pi}_{i}} , \nonumber
\end{eqnarray}
from which, along with $\hat{\theta}_{k,\mathrm{KNN}}$, $\hat{\beta}_{jk,\mathrm{KNN}}$ and $\hat{\gamma}_{jk,\mathrm{KNN}}$, one derives the estimates of the variances of
the proposed KNN imputation estimators.

To obtain estimates of covariances, 
we need to estimate also the quantities $\psi^2_{1212}$, $\psi^2_{112}$, $\psi^2_{213}$, $\psi^2_{12}$, $\psi^2_{113}$, $\psi^2_{223}$ and $\psi^2_{1223}$  given in Appendix 2.
However, estimates of such quantities are similar to those given above for
${\omega}_k^2, {\omega}_{jk}^2$ and ${\eta}_{jk}^2$. For example,
\begin{eqnarray}
\hat{\psi}_{1212}^2 &=& \frac{K + 1}{nK}\sum_{i=1}^{n}\mathrm{I}(c_1 \le T_i < c_2)\tilde{\rho}_{1i}\tilde{\rho}_{2i}\left(1 - \tilde{\pi}_{i}\right) \nonumber \\ 
&& + \: \frac{1}{n}\sum_{i = 1}^{n}\frac{\mathrm{I}(c_1 \le T_i < c_2)\tilde{\rho}_{1i}
\tilde{\rho}_{2i}\left(1 - \tilde{\pi}_{i}\right)^2}{\tilde{\pi}_{i}}. \nonumber
\end{eqnarray}


Of course, there are other possible approaches to obtain variance and covariance estimates. 
For instance, one could resort to a standard bootstrap procedure.
From the original observations $(T_i,A_i,D_i,V_i)$, $i = 1,\ldots,n$, consider $B$ bootstrap samples $(T_i^{*b},A_i^{*b},D_i^{*b},V_i^{*b})$, $b = 1,\ldots,B$, and $i = 1,\ldots,n$. For the $b$-th sample, compute the bootstrap estimates $\widehat{\TCF}_{1,\mathrm{KNN}}^{*b}(c_1)$, $\widehat{\TCF}_{2,\mathrm{KNN}}^{*b}(c_1,c_2)$ and $\widehat{\TCF}_{3,\mathrm{KNN}}^{*b}(c_2)$ as 
\begin{align}
\widehat{\TCF}_{1,\mathrm{KNN}}^{*b}(c_1) &=  \frac{\sum\limits_{i=1}^{n}\mathrm{I}(T_i^{*b} < c_1)\left[V_i^{*b} D^{*b}_{1i} + (1 - V_i^{*b}) \hat{\rho}_{1i,K}^{*b}\right]}{\sum\limits_{i=1}^{n}\left[V_i^{*b} D^{*b}_{1i} + (1 - V_i^{*b}) \hat{\rho}_{1i,K}^{*b}\right]} \nonumber, \\
\widehat{\TCF}_{2,\mathrm{KNN}}^{*b}(c_1,c_2) &= \frac{\sum\limits_{i=1}^{n}\mathrm{I}(c_1 \le T_i^{*b} < c_2)\left[V_i^{*b} D^{*b}_{2i} + (1 - V_i^{*b}) \hat{\rho}_{2i,K}^{*b}\right]}{\sum\limits_{i=1}^{n}\left[V_i^{*b} D^{*b}_{2i} + (1 - V_i^{*b}) \hat{\rho}_{2i,K}^{*b}\right]}, \nonumber \label{boot:est:knn} \displaybreak[3] \\
\widehat{\TCF}_{3,\mathrm{KNN}}^{*b}(c_2) &=  \frac{\sum\limits_{i=1}^{n}\mathrm{I}(T_i^{*b} \ge c_2)\left[V_i^{*b} D^{*b}_{3i} + (1 - V_i^{*b}) \hat{\rho}_{3i,K}^{*b}\right]}{\sum\limits_{i=1}^{n}\left[V_i^{*b} D^{*b}_{3i} + (1 - V_i^{*b}) \hat{\rho}_{3i,K}^{*b}\right]}, \nonumber
\end{align}
where $\hat{\rho}_{ki,K}^{*b}$, $k = 1,2,3$, denote the KNN imputation values for missing labels $D^{*b}_{ki}$ in the bootstrap sample. Then, the bootstrap estimator of the variance of $\widehat{\TCF}_{k,\mathrm{KNN}}(c_1,c_2)$ is
\begin{equation}
\widehat{\V}(\widehat{\TCF}_{k,\mathrm{KNN}}(c_1,c_2)) = \frac{1}{B-1}\sum_{b=1}^{B}\left(\widehat{\TCF}_{k,\mathrm{KNN}}^{*b}(c_1,c_2) - \widehat{\TCF}_{k,\mathrm{KNN}}^{*}(c_1,c_2)\right)^2 \nonumber \label{var:boot:est:knn},
\end{equation}
where $\widehat{\TCF}_{k,\mathrm{KNN}}^{*}(c_1,c_2)$ is the mean of the $B$ bootstrap estimates $\widehat{\TCF}_{k,\mathrm{KNN}}^{*b}(c_1,c_2)$. 
More generally, the bootstrap estimate of the covariance matrix $\Xi$ is 
\begin{equation}
\widehat{\Xi}_B = \frac{1}{B-1}\left(\widehat{\TCF}_{\mathrm{KNN}}^{*B}(c_1,c_2) - \widehat{\TCF}_{\mathrm{KNN}}^{*}(c_1,c_2)\right)\left(\widehat{\TCF}_{\mathrm{KNN}}^{*B}(c_1,c_2) - \widehat{\TCF}_{\mathrm{KNN}}^{*}(c_1,c_2)\right)^\top,
\nonumber \label{covar:boot:est:knn}
\end{equation}
where $\widehat{\TCF}_{\mathrm{KNN}}^{*B}(c_1,c_2)$ is a $B \times 3$ matrix, whose element in the $b$--th row and the $k$--th column corresponds to $\widehat{\TCF}_{k,\mathrm{KNN}}^{*b}(c_1,c_2)$, and $\widehat{\TCF}_{\mathrm{KNN}}^{*}(c_1,c_2)$ is a column vector that consist of the means of the $B$ bootstrap estimates $\widehat{\TCF}_{k,\mathrm{KNN}}^{*b}(c_1,c_2)$, $k=1,2,3$. 

\section{Simulation studies}
\label{s:simulation}
In this section, the ability of KNN method to estimate TCF$_1$, TCF$_2$ and TCF$_3$ is evaluated by using Monte Carlo experiments. 
We also compare the proposed method with partially parametric approaches, i.e., FI, MSI, IPW and SPE approaches. As already mentioned,
partially parametric bias-corrected estimators of TCF$_1$, TCF$_2$ and TCF$_3$ require  parametric regression models to estimate $\rho_{ki} = \Pro(D_{ki} = 1|T_i, A_i)$, or $\pi_i = \Pro(V_i = 1 | T_i,A_i)$, or both. A wrong specification of such models may affect the estimators. Therefore, in the simulation study we consider two scenarios: in the parametric estimation process,
\begin{enumerate}
\item[(i)] the disease model and the verification model are both correctly specified;
\item[(ii)] the disease model and the verification model are both misspecified.
\end{enumerate}
In both scenarios,
we execute $5000$ Monte Carlo runs at each setting; we set three sample sizes, i.e., $250$, $500$ and $1000$ in scenario (i) and a sample size of $1000$ in scenario (ii).

We consider KNN  estimators  based on the Euclidean distance, with  $K = 1$ and $K=3$.
This in light of the discussion in Section 3.4 and some results of a preliminary simulation study presented in Section S1, Supplementary Material. 
In such study, we compared the behavior of the KNN estimators for several choices of the distance measure (Euclidean, Manhattan, Canberra and Mahalanobis) and  the size of the neighborhood ($K = 1,3,5,10,20$).

\subsection{Correctly specified parametric models}
The true disease $D$ is generated by a trinomial random vector $(D_{1},D_{2},D_{3})$, such that $D_{k}$ is a Bernoulli random variable with success probability $\theta_k$, $k=1,2,3$. We set $\theta_1 = 0.4, \theta_2 = 0.35$ and $\theta_3 = 0.25$. The continuous test result $T$ and a covariate $A$ are generated from the following conditional models
\[
T,A |D_{k} \sim \mathcal{N}_2 \left(\mu_k, \Sigma\right), \qquad k = 1,2,3,
\]
where $\mu_k = (2 k, k)^\top$ and 
\[
\Sigma = \left(\begin{array}{c c}
\sigma^2_{T|D} & \sigma_{T,A|D} \\
\sigma_{T,A|D} & \sigma^2_{A|D}
\end{array}\right).
\]
We consider three different values for $\Sigma$, specifically 
\[
\left(\begin{array}{c c}
1.75 & 0.1 \\
0.1 & 2.5
\end{array}\right) , \qquad 
\left(\begin{array}{c c}
2.5 & 1.5 \\
1.5 & 2.5
\end{array}\right) , \qquad 
\left(\begin{array}{c c}
5.5 & 3 \\
3 & 2.5
\end{array}\right),
\]
giving rise to a correlation between $T$ and $A$ equal to $0.36, 0.69$ and $0.84$, respectively. Values chosen for $\Sigma$ give rise to true VUS values ranging from 0.7175 to 0.4778.  The verification status $V$ is generated by the following model
\[
\mathrm{logit}\left\{\Pro(V = 1|T,A)\right\} = \delta_0 + \delta_1 T + \delta_2 A,
\]
where we fix $\delta_0 = 0.5, \delta_1 = -0.3$ and $\delta_2 = 0.75$. This choice corresponds to a verification rate of about $0.65$. We consider six pairs of cut points $(c_1,c_2)$, i.e., $(2,4), (2,5)$, $(2,7), (4,5), (4,7)$ and $(5,7)$.
Since the conditional distribution of $T$ given $D_{k}$ is the normal distribution, the true parameters values are 
\begin{eqnarray}
{\TCF}_{1}(c_1) &=& \Phi \left(\frac{c_1 - 2}{\sigma_{T|D}}\right) \nonumber, \\
{\TCF}_{2}(c_1,c_2) &=& \Phi \left(\frac{c_2 - 4}{\sigma_{T|D}}\right) - \Phi \left(\frac{c_1 - 4}{\sigma_{T|D}}\right), \nonumber\\
{\TCF}_{3}(c_2) &=& 1 - \Phi \left(\frac{c_2 - 6}{\sigma_{T|D}}\right), \nonumber
\end{eqnarray}
where $\Phi(\cdot)$ denotes the cumulative distribution function of the standard normal random variable.

In this set--up,  FI, MSI, IPW and SPE estimators are computed under correct working models for both the disease and the verification processes. Therefore, the conditional verification probabilities $\pi_i$ are estimated from a logistic model for $V$ given $T$ and $A$ with logit link. Under our data--generating process, the true conditional disease model is a multinomial logistic model
\[
\mathrm{Pr}(D_{k} = 1|T,A) = \frac{\exp \left(\tau_{0k} + \tau_{1k} T + \tau_{2k} A\right)}{1 + \exp \left(\tau_{01} + \tau_{11} T + \tau_{21} A\right) + \exp \left(\tau_{02} + \tau_{12} T + \tau_{22} A\right)},
\]
for suitable $\tau_{0k},\tau_{1k},\tau_{2k}$, where $k = 1,2$.

Tables \ref{tab:res11}--\ref{tab:res13} show Monte Carlo means and standard deviations of the estimators for the three true class factions. Results concern the estimators FI, MSI, IPW, SPE, and the KNN estimator with $K = 1$ and $K = 3$ computed using the Euclidean distance. 
Also, the estimated standard deviations are shown in the tables. The estimates are obtained by using asymptotic results. To estimate standard deviations of KNN estimators, we use the KNN procedure discussed in Section 4, with $\bar{K} = 2$. Each table refers to a choosen value for $\Sigma$. The sample size is  $250$. The results for sample sizes $500$ and $1000$ are presented in Section S2 of Supplementary Material.

As expected, the parametric approaches work well when both models for $\rho_k(t,a)$ and $\pi(t,a)$ are correctly specified. FI and MSI estimators seem to be the most efficient ones, whereas the IPW approach seems to provide less powerful estimators, in general.
The new proposals (1NN and 3NN estimators) yield also good results, comparable, in terms of bias and standard deviation, to those of the parametric competitors. 
Moreover,  estimators 1NN and 3NN seem to achieve similar performances, 
and the results about estimated standard deviations of KNN estimators seem to show the effectiveness of the procedure discussed in Section 4.


Finally, some results of simulation experiments performed to explore the effect of a multidimensional vector of auxiliary covariates are given in Section S3, Supplementary Material. A vector $A$ of dimension 3 is employed. The results in Table 16, Supplementary Material,  show that KNN estimators still behave satisfactorily.

\begin{table}[htbp]
\caption{Monte Carlo means, Monte Carlo standard deviations and estimated standard deviations of the estimators for true class fractions, in case of sample size equals to $250$. The first value of $\Sigma$ is considered. ``True'' denotes the true parameter value.}
\begin{center}
\begin{footnotesize}
\begin{tabular}{c c c c c c c c c c}
\toprule
	& TCF$_1$ & TCF$_2$ & TCF$_3$ & MC.sd$_1$ & MC.sd$_2$ & MC.sd$_3$ & asy.sd$_1$ & asy.sd$_2$ & asy.sd$_3$ \\
\midrule
	& \multicolumn{9}{c}{cut points $ = (2,4)$} \\
\hdashline
	True & 0.5000 & 0.4347 & 0.9347 &  &  &  &  &  &  \\
	FI & 0.5005 & 0.4348 & 0.9344 & 0.0537 & 0.0484 & 0.0269 & 0.0440 & 0.0398 & 0.0500 \\ 
	MSI & 0.5005 & 0.4346 & 0.9342 & 0.0550 & 0.0547 & 0.0320 & 0.0465 & 0.0475 & 0.0536 \\ 
	IPW & 0.4998 & 0.4349 & 0.9341 & 0.0722 & 0.0727 & 0.0372 & 0.0688 & 0.0702 & 0.0420 \\ 
	SPE & 0.5010 & 0.4346 & 0.9344 & 0.0628 & 0.0659 & 0.0364 & 0.0857 & 0.0637 & 0.0363 \\ 
	1NN & 0.4989 & 0.4334 & 0.9331 & 0.0592 & 0.0665 & 0.0387 & 0.0555 & 0.0626 & 0.0382 \\ 
	3NN & 0.4975 & 0.4325 & 0.9322 & 0.0567 & 0.0617 & 0.0364 & 0.0545 & 0.0608 & 0.0372 \\
\midrule
	& \multicolumn{9}{c}{cut points $ = (2,5)$} \\
\hdashline
	True & 0.5000 & 0.7099 & 0.7752 &  &  &  &  &  &  \\ 
  	FI & 0.5005 & 0.7111 & 0.7761 & 0.0537 & 0.0461 & 0.0534 & 0.0440 & 0.0400 & 0.0583 \\ 
	MSI & 0.5005 & 0.7104 & 0.7756 & 0.0550 & 0.0511 & 0.0566 & 0.0465 & 0.0467 & 0.0626 \\ 
  	IPW & 0.4998 & 0.7108 & 0.7750 & 0.0722 & 0.0701 & 0.0663 & 0.0688 & 0.0667 & 0.0713 \\ 
  	SPE & 0.5010 & 0.7106 & 0.7762 & 0.0628 & 0.0619 & 0.0627 & 0.0857 & 0.0604 & 0.0611 \\
	1NN &  0.4989 & 0.7068 & 0.7738 & 0.0592 & 0.0627 & 0.0652 & 0.0555 & 0.0591 & 0.0625 \\ 
	3NN & 0.4975 & 0.7038 & 0.7714 & 0.0567 & 0.0576 & 0.0615 & 0.0545 & 0.0574 & 0.0610 \\
\midrule
	& \multicolumn{9}{c}{cut points $ = (2,7)$} \\
\hdashline
	True & 0.5000 & 0.9230 & 0.2248 &  &  &  &  &  &  \\ 
 	FI & 0.5005 & 0.9229 & 0.2240 & 0.0537 & 0.0236 & 0.0522 & 0.0440 & 0.0309 & 0.0428 \\ 
 	MSI & 0.5005 & 0.9231 & 0.2243 & 0.0550 & 0.0285 & 0.0531 & 0.0465 & 0.0353 & 0.0443 \\ 
  	IPW & 0.4998 & 0.9238 & 0.2222 & 0.0722 & 0.0374 & 0.0765 & 0.0688 & 0.0360 & 0.0728 \\ 
  	SPE & 0.5010 & 0.9236 & 0.2250 & 0.0628 & 0.0362 & 0.0578 & 0.0857 & 0.0348 & 0.0573 \\ 
	1NN & 0.4989 & 0.9201 & 0.2233 & 0.0592 & 0.0372 & 0.0577 & 0.0555 & 0.0366 & 0.0570 \\ 
  	3NN & 0.4975 & 0.9177 & 0.2216 & 0.0567 & 0.0340 & 0.0558 & 0.0545 & 0.0355 & 0.0563 \\
\midrule
	& \multicolumn{9}{c}{cut points $ = (4,5)$} \\
\hdashline
	True & 0.9347 & 0.2752 & 0.7752 &  &  &  &  &  &  \\
  	FI & 0.9347 & 0.2763 & 0.7761 & 0.0245 & 0.0412 & 0.0534 & 0.0179 & 0.0336 & 0.0583 \\ 
  	MSI & 0.9348 & 0.2758 & 0.7756 & 0.0271 & 0.0471 & 0.0566 & 0.0220 & 0.0404 & 0.0626 \\ 
  	IPW & 0.9350 & 0.2758 & 0.7750 & 0.0421 & 0.0693 & 0.0663 & 0.0391 & 0.0651 & 0.0713 \\ 
  	SPE & 0.9353 & 0.2761 & 0.7762 & 0.0386 & 0.0590 & 0.0627 & 0.0377 & 0.0568 & 0.0611 \\
  	1NN & 0.9322 & 0.2734 & 0.7738 & 0.0374 & 0.0572 & 0.0652 & 0.0342 & 0.0553 & 0.0625 \\ 
  	3NN & 0.9303 & 0.2712 & 0.7714 & 0.0328 & 0.0526 & 0.0615 & 0.0332 & 0.0538 & 0.0610 \\
\midrule
	& \multicolumn{9}{c}{cut points $ = (4,7)$} \\
\hdashline
	True & 0.9347 & 0.4883 & 0.2248 &  &  &  &  &  &  \\
  	FI & 0.9347 & 0.4881 & 0.2240 & 0.0245 & 0.0541 & 0.0522 & 0.0179 & 0.0444 & 0.0428 \\ 
  	MSI & 0.9348 & 0.4885 & 0.2243 & 0.0271 & 0.0576 & 0.0531 & 0.0220 & 0.0495 & 0.0443 \\ 
  	IPW & 0.9350 & 0.4889 & 0.2222 & 0.0421 & 0.0741 & 0.0765 & 0.0391 & 0.0713 & 0.0728 \\ 
  	SPE & 0.9353 & 0.4890 & 0.2250 & 0.0386 & 0.0674 & 0.0578 & 0.0377 & 0.0646 & 0.0573 \\  
  	1NN & 0.9322 & 0.4867 & 0.2233 & 0.0374 & 0.0680 & 0.0577 & 0.0342 & 0.0633 & 0.0570 \\ 
  	3NN & 0.9303 & 0.4852 & 0.2216 & 0.0328 & 0.0630 & 0.0558 & 0.0332 & 0.0615 & 0.0563 \\
\midrule  	
	& \multicolumn{9}{c}{cut points $ = (5,7)$} \\
\hdashline
	True & 0.9883 & 0.2132 & 0.2248 &  &  &  &  &  &  \\
  	FI & 0.9879 & 0.2118 & 0.2240 & 0.0075 & 0.0435 & 0.0522 & 0.0055 & 0.0336 & 0.0428 \\ 
  	MSI & 0.9882 & 0.2127 & 0.2243 & 0.0096 & 0.0467 & 0.0531 & 0.0084 & 0.0388 & 0.0443 \\ 
  	IPW & 0.9887 & 0.2130 & 0.2222 & 0.0193 & 0.0653 & 0.0765 & 0.0177 & 0.0618 & 0.0728 \\ 
  	SPE & 0.9888 & 0.2130 & 0.2250 & 0.0191 & 0.0571 & 0.0578 & 0.0184 & 0.0554 & 0.0573 \\ 
  	1NN & 0.9868 & 0.2133 & 0.2233 & 0.0177 & 0.0567 & 0.0577 & 0.0172 & 0.0532 & 0.0570 \\ 
  	3NN & 0.9860 & 0.2139 & 0.2216 & 0.0151 & 0.0519 & 0.0558 & 0.0168 & 0.0516 & 0.0563 \\  	
\bottomrule
\end{tabular}
\end{footnotesize}
\end{center}
\label{tab:res11}
\end{table}

\begin{table}[htbp]
\caption{Monte Carlo means, Monte Carlo standard deviations and estimated standard deviations of the estimators for true class fractions, in case of sample size equals to $250$. The second value of $\Sigma$ is considered. ``True'' denotes the true parameter value.}
\begin{center}
\begin{footnotesize}
\begin{tabular}{c c c c c c c c c c}
\toprule
	& TCF$_1$ & TCF$_2$ & TCF$_3$ & MC.sd$_1$ & MC.sd$_2$ & MC.sd$_3$ & asy.sd$_1$ & asy.sd$_2$ & asy.sd$_3$ \\
\midrule
	& \multicolumn{9}{c}{cut points $ = (2,4)$} \\
\hdashline
	True & 0.5000 & 0.3970 & 0.8970 &  &  &  &  &  &  \\ 
  	FI & 0.4999 & 0.3974 & 0.8973 & 0.0503 & 0.0421 & 0.0362 & 0.0432 & 0.0352 & 0.0466 \\ 
  	MSI & 0.5000 & 0.3975 & 0.8971 & 0.0521 & 0.0497 & 0.0416 & 0.0461 & 0.0451 & 0.0515 \\ 
  	IPW & 0.4989 & 0.3990 & 0.8971 & 0.0663 & 0.0685 & 0.0534 & 0.0647 & 0.0681 & 0.0530 \\ 
  	SPE & 0.5004 & 0.3980 & 0.8976 & 0.0570 & 0.0619 & 0.0516 & 0.0563 & 0.0620 & 0.0493 \\
    1NN & 0.4982 & 0.3953 & 0.8976 & 0.0587 & 0.0642 & 0.0537 & 0.0561 & 0.0618 & 0.0487 \\ 
    3NN & 0.4960 & 0.3933 & 0.8970 & 0.0556 & 0.0595 & 0.0494 & 0.0548 & 0.0600 & 0.0472 \\
\midrule
	& \multicolumn{9}{c}{cut points $ = (2,5)$} \\
\hdashline
	True & 0.5000 & 0.6335 & 0.7365 &  &  &  &  &  &  \\ 
  	FI & 0.4999 & 0.6337 & 0.7395 & 0.0503 & 0.0436 & 0.0583 & 0.0432 & 0.0379 & 0.0554 \\ 
  	MSI & 0.5000 & 0.6330 & 0.7385 & 0.0521 & 0.0508 & 0.0613 & 0.0461 & 0.0469 & 0.0612 \\ 
  	IPW & 0.4989 & 0.6335 & 0.7386 & 0.0663 & 0.0676 & 0.0728 & 0.0647 & 0.0663 & 0.0745 \\ 
  	SPE & 0.5004 & 0.6333 & 0.7390 & 0.0570 & 0.0622 & 0.0682 & 0.0563 & 0.0612 & 0.0673 \\ 
    1NN & 0.4982 & 0.6304 & 0.7400 & 0.0587 & 0.0645 & 0.0721 & 0.0561 & 0.0615 & 0.0672 \\ 
    3NN & 0.4960 & 0.6283 & 0.7396 & 0.0556 & 0.0600 & 0.0670 & 0.0548 & 0.0597 & 0.0654 \\ 
\midrule
	& \multicolumn{9}{c}{cut points $ = (2,7)$} \\
\hdashline
	True & 0.5000 & 0.8682 & 0.2635 &  &  &  &  &  &  \\ 
  	FI & 0.4999 & 0.8676 & 0.2655 & 0.0503 & 0.0316 & 0.0560 & 0.0432 & 0.0294 & 0.0478 \\ 
  	MSI & 0.5000 & 0.8678 & 0.2660 & 0.0521 & 0.0374 & 0.0583 & 0.0461 & 0.0364 & 0.0512 \\ 
  	IPW & 0.4989 & 0.8682 & 0.2669 & 0.0663 & 0.0507 & 0.0698 & 0.0647 & 0.0484 & 0.0692 \\ 
  	SPE & 0.5004 & 0.8681 & 0.2663 & 0.0570 & 0.0476 & 0.0608 & 0.0563 & 0.0459 & 0.0600 \\
    1NN & 0.4982 & 0.8672 & 0.2672 & 0.0587 & 0.0495 & 0.0629 & 0.0561 & 0.0458 & 0.0609 \\ 
    3NN & 0.4960 & 0.8657 & 0.2671 & 0.0556 & 0.0452 & 0.0610 & 0.0548 & 0.0442 & 0.0601 \\
\midrule
	& \multicolumn{9}{c}{cut points $ = (4,5)$} \\
\hdashline
	True & 0.8970 & 0.2365 & 0.7365 &  &  &  &  &  &  \\ 
  	FI & 0.8980 & 0.2363 & 0.7395 & 0.0284 & 0.0367 & 0.0583 & 0.0239 & 0.0301 & 0.0554 \\ 
  	MSI & 0.8976 & 0.2356 & 0.7385 & 0.0318 & 0.0437 & 0.0613 & 0.0292 & 0.0386 & 0.0612 \\ 
  	IPW & 0.8975 & 0.2345 & 0.7386 & 0.0377 & 0.0594 & 0.0728 & 0.0373 & 0.0578 & 0.0745 \\ 
  	SPE & 0.8974 & 0.2353 & 0.7390 & 0.0364 & 0.0529 & 0.0682 & 0.0361 & 0.0522 & 0.0673 \\ 
    1NN & 0.8958 & 0.2352 & 0.7400 & 0.0388 & 0.0540 & 0.0721 & 0.0373 & 0.0524 & 0.0672 \\ 
    3NN & 0.8946 & 0.2350 & 0.7396 & 0.0362 & 0.0502 & 0.0670 & 0.0361 & 0.0510 & 0.0654 \\
\midrule
	& \multicolumn{9}{c}{cut points $ = (4,7)$} \\
\hdashline
	True & 0.8970 & 0.4711 & 0.2635 &  &  &  &  &  &  \\ 
  	FI & 0.8980 & 0.4703 & 0.2655 & 0.0284 & 0.0512 & 0.0560 & 0.0239 & 0.0413 & 0.0478 \\ 
  	MSI & 0.8976 & 0.4703 & 0.2660 & 0.0318 & 0.0561 & 0.0583 & 0.0292 & 0.0490 & 0.0512 \\ 
  	IPW & 0.8975 & 0.4692 & 0.2669 & 0.0377 & 0.0693 & 0.0698 & 0.0373 & 0.0679 & 0.0692 \\ 
  	SPE & 0.8974 & 0.4701 & 0.2663 & 0.0364 & 0.0638 & 0.0608 & 0.0361 & 0.0629 & 0.0600 \\ 
    1NN & 0.8958 & 0.4719 & 0.2672 & 0.0388 & 0.0666 & 0.0629 & 0.0373 & 0.0630 & 0.0609 \\ 
    3NN & 0.8946 & 0.4724 & 0.2671 & 0.0362 & 0.0627 & 0.0610 & 0.0361 & 0.0611 & 0.0601 \\ 
\midrule  	
	& \multicolumn{9}{c}{cut points $ = (5,7)$} \\
\hdashline
	True & 0.9711 & 0.2347 & 0.2635 &  &  &  &  &  &  \\ 
  	FI & 0.9710 & 0.2339 & 0.2655 & 0.0124 & 0.0407 & 0.0560 & 0.0104 & 0.0336 & 0.0478 \\ 
  	MSI & 0.9709 & 0.2348 & 0.2660 & 0.0166 & 0.0461 & 0.0583 & 0.0156 & 0.0412 & 0.0512 \\ 
  	IPW & 0.9709 & 0.2347 & 0.2669 & 0.0204 & 0.0568 & 0.0698 & 0.0202 & 0.0562 & 0.0692 \\ 
  	SPE & 0.9709 & 0.2348 & 0.2663 & 0.0202 & 0.0531 & 0.0608 & 0.0199 & 0.0524 & 0.0600 \\ 
    1NN & 0.9701 & 0.2368 & 0.2672 & 0.0217 & 0.0549 & 0.0629 & 0.0213 & 0.0533 & 0.0609 \\ 
    3NN & 0.9695 & 0.2375 & 0.2671 & 0.0200 & 0.0519 & 0.0610 & 0.0206 & 0.0517 & 0.0601 \\ 	
\bottomrule
\end{tabular}
\end{footnotesize}
\end{center}
\label{tab:res12}
\end{table}

\begin{table}[htbp]
\caption{Monte Carlo means, Monte Carlo standard deviations and estimated standard deviations of the estimators for true class fractions, in case of sample size equals to $250$. The third value of $\Sigma$ is considered. ``True'' denotes the true parameter value.}
\begin{center}
\begin{footnotesize}
\begin{tabular}{c c c c c c c c c c}
\toprule
	& TCF$_1$ & TCF$_2$ & TCF$_3$ & MC.sd$_1$ & MC.sd$_2$ & MC.sd$_3$ & asy.sd$_1$ & asy.sd$_2$ & asy.sd$_3$ \\
\midrule
	& \multicolumn{9}{c}{cut points $ = (2,4)$} \\
\hdashline
	True & 0.5000 & 0.3031 & 0.8031 &  &  &  &  &  &  \\ 
  	FI & 0.5009 & 0.3031 & 0.8047 & 0.0488 & 0.0344 & 0.0495 & 0.0418 & 0.0284 & 0.0467 \\ 
  	MSI & 0.5005 & 0.3032 & 0.8045 & 0.0515 & 0.0448 & 0.0544 & 0.0460 & 0.0410 & 0.0542 \\ 
  	IPW & 0.5015 & 0.3030 & 0.8043 & 0.0624 & 0.0632 & 0.0649 & 0.0618 & 0.0620 & 0.0640 \\ 
  	SPE & 0.5007 & 0.3034 & 0.8043 & 0.0565 & 0.0576 & 0.0628 & 0.0564 & 0.0574 & 0.0614 \\ 
    1NN & 0.4997 & 0.3021 & 0.8047 & 0.0592 & 0.0602 & 0.0682 & 0.0571 & 0.0584 & 0.0621 \\ 
    3NN & 0.4984 & 0.3018 & 0.8043 & 0.0561 & 0.0565 & 0.0632 & 0.0556 & 0.0566 & 0.0601 \\ 
\midrule
	& \multicolumn{9}{c}{cut points $ = (2,5)$} \\
\hdashline
	True & 0.5000 & 0.4682 & 0.6651 &  &  &  &  &  &  \\ 
  	FI & 0.5009 & 0.4692 & 0.6668 & 0.0488 & 0.0384 & 0.0616 & 0.0418 & 0.0323 & 0.0536 \\ 
  	MSI & 0.5005 & 0.4687 & 0.6666 & 0.0515 & 0.0495 & 0.0658 & 0.0460 & 0.0455 & 0.0610 \\ 
  	IPW & 0.5015 & 0.4681 & 0.6670 & 0.0624 & 0.0671 & 0.0753 & 0.0618 & 0.0670 & 0.0743 \\ 
  	SPE & 0.5007 & 0.4690 & 0.6665 & 0.0565 & 0.0624 & 0.0721 & 0.0564 & 0.0622 & 0.0704 \\ 
    1NN & 0.4997 & 0.4676 & 0.6668 & 0.0592 & 0.0661 & 0.0780 & 0.0571 & 0.0634 & 0.0717 \\ 
    3NN & 0.4984 & 0.4670 & 0.6666 & 0.0561 & 0.0619 & 0.0729 & 0.0556 & 0.0614 & 0.0695 \\ 
\midrule
	& \multicolumn{9}{c}{cut points $ = (2,7)$} \\
\hdashline
	True & 0.5000 & 0.7027 & 0.3349 &  &  &  &  &  &  \\ 
  	FI & 0.5009 & 0.7030 & 0.3358 & 0.0488 & 0.0375 & 0.0595 & 0.0418 & 0.0318 & 0.0501 \\ 
  	MSI & 0.5005 & 0.7027 & 0.3360 & 0.0515 & 0.0474 & 0.0637 & 0.0460 & 0.0435 & 0.0563 \\ 
  	IPW & 0.5015 & 0.7026 & 0.3366 & 0.0624 & 0.0625 & 0.0730 & 0.0618 & 0.0618 & 0.0716 \\ 
  	SPE & 0.5007 & 0.7032 & 0.3362 & 0.0565 & 0.0591 & 0.0677 & 0.0564 & 0.0583 & 0.0657 \\ 
    1NN & 0.4997 & 0.7024 & 0.3366 & 0.0592 & 0.0633 & 0.0712 & 0.0571 & 0.0592 & 0.0675 \\ 
    3NN & 0.4984 & 0.7016 & 0.3362 & 0.0561 & 0.0590 & 0.0680 & 0.0556 & 0.0572 & 0.0660 \\ 
\midrule
	& \multicolumn{9}{c}{cut points $ = (4,5)$} \\
\hdashline
	True & 0.8031 & 0.1651 & 0.6651 &  &  &  &  &  &  \\ 
  	FI & 0.8042 & 0.1660 & 0.6668 & 0.0383 & 0.0277 & 0.0616 & 0.0323 & 0.0231 & 0.0536 \\ 
  	MSI & 0.8037 & 0.1655 & 0.6666 & 0.0415 & 0.0372 & 0.0658 & 0.0380 & 0.0333 & 0.0610 \\ 
  	IPW & 0.8039 & 0.1651 & 0.6670 & 0.0473 & 0.0503 & 0.0753 & 0.0473 & 0.0493 & 0.0743 \\ 
  	SPE & 0.8036 & 0.1655 & 0.6665 & 0.0456 & 0.0465 & 0.0721 & 0.0458 & 0.0455 & 0.0704 \\ 
    1NN & 0.8032 & 0.1655 & 0.6668 & 0.0487 & 0.0481 & 0.0780 & 0.0472 & 0.0466 & 0.0717 \\ 
    3NN & 0.8020 & 0.1651 & 0.6666 & 0.0460 & 0.0450 & 0.0729 & 0.0457 & 0.0451 & 0.0695 \\ 
\midrule
	& \multicolumn{9}{c}{cut points $ = (4,7)$} \\
\hdashline
	True & 0.8031 & 0.3996 & 0.3349 &  &  &  &  &  &  \\ 
 	FI & 0.8042 & 0.3999 & 0.3358 & 0.0383 & 0.0426 & 0.0595 & 0.0323 & 0.0349 & 0.0501 \\ 
  	MSI & 0.8037 & 0.3995 & 0.3360 & 0.0415 & 0.0522 & 0.0637 & 0.0380 & 0.0463 & 0.0563 \\ 
  	IPW & 0.8039 & 0.3996 & 0.3366 & 0.0473 & 0.0658 & 0.0730 & 0.0473 & 0.0645 & 0.0716 \\ 
  	SPE & 0.8036 & 0.3998 & 0.3362 & 0.0456 & 0.0618 & 0.0677 & 0.0458 & 0.0606 & 0.0657 \\ 
    1NN & 0.8032 & 0.4003 & 0.3366 & 0.0487 & 0.0660 & 0.0712 & 0.0472 & 0.0619 & 0.0675 \\ 
    3NN & 0.8020 & 0.3998 & 0.3362 & 0.0460 & 0.0617 & 0.0680 & 0.0457 & 0.0600 & 0.0660 \\ 
\midrule  	
	& \multicolumn{9}{c}{cut points $ = (5,7)$} \\
\hdashline
	True & 0.8996 & 0.2345 & 0.3349 &  &  &  &  &  &  \\ 
  	FI & 0.9003 & 0.2338 & 0.3358 & 0.0266 & 0.0351 & 0.0595 & 0.0224 & 0.0292 & 0.0501 \\ 
  	MSI & 0.9004 & 0.2340 & 0.3360 & 0.0308 & 0.0443 & 0.0637 & 0.0285 & 0.0398 & 0.0563 \\ 
  	IPW & 0.9005 & 0.2345 & 0.3366 & 0.0355 & 0.0555 & 0.0730 & 0.0353 & 0.0550 & 0.0716 \\ 
  	SPE & 0.9004 & 0.2342 & 0.3362 & 0.0349 & 0.0523 & 0.0677 & 0.0346 & 0.0517 & 0.0657 \\ 
    1NN & 0.9000 & 0.2348 & 0.3366 & 0.0373 & 0.0556 & 0.0712 & 0.0361 & 0.0531 & 0.0675 \\ 
    3NN & 0.8992 & 0.2346 & 0.3362 & 0.0349 & 0.0520 & 0.0680 & 0.0349 & 0.0515 & 0.0660 \\ 
\bottomrule
\end{tabular}
\end{footnotesize}
\end{center}
\label{tab:res13}
\end{table}

\subsection{Misspecified models}
We start from two independent random variables $Z_1 \sim \mathcal{N}(0,0.5)$ and $Z_2 \sim \mathcal{N}(0,0.5)$. The true conditional disease $D$ is generated by a trinomial random vector $(D_1,D_2,D_3)$ such that
\[
D_1 = \left\{\begin{array}{r l}
1 & \mathrm{ if } \, Z_1 + Z_2 \le h_1 \\
0 & \mathrm{ otherwise}
\end{array}\right. ,
\quad 
D_2 = \left\{\begin{array}{r l}
1 & \mathrm{ if } \, h_1 < Z_1 + Z_2 \le h_2 \\
0 & \mathrm{ otherwise}
\end{array}\right. ,
\quad
D_3 = \left\{\begin{array}{r l}
1 & \mathrm{ if } \, Z_1 + Z_2 > h_2 \\
0 & \mathrm{ otherwise}
\end{array}\right. .
\]
Here, $h_1$ and $h_2$ are two thresholds. We choose $h_1$ and $h_2$ to make $\theta_1 = 0.4$ and $\theta_3 = 0.25$. The continuous test results $T$ and the covariate $A$ are generated to be related to $D$ through $Z_1$ and $Z_2$. More precisely,
\[
T  = \alpha(Z_1 + Z_2) + \varepsilon_1, \qquad A = Z_1 + Z_2 + \varepsilon_2,
\]
where $\varepsilon_1$ and $\varepsilon_2$ are two independent normal random variables with mean $0$ and the common variance $0.25$. The verification status $V$ is simulated by the following logistic model
\[
\mathrm{logit}\left\{Pr(V = 1|T,A)\right\} = -1.5 - 0.35 T - 1.5A.
\]
Under this model, the verification rate is roughly $0.276$. This has led us to the choice of $n = 1000$. For the cut-point, we consider six pairs $(c_1,c_2)$, i.e., $(-1.0, -0.5)$, $(-1.0,0.7)$, $(-1.0, 1.3)$, $(-0.5,0.7)$, $(-0.5,1.3)$ and $(0.7, 1.3)$. Within this set--up, we determine the true values of TCF's as follows:
\begin{eqnarray}
{\TCF}_{1}(c_1) &=& \frac{1}{\Phi(h_1)} \int_{-\infty}^{h_1}\Phi\left(\frac{c_1 - \alpha z}{\sqrt{0.25}}\right)\phi(z)\ud z\nonumber, \\
{\TCF}_{2}(c_1,c_2) &=& \frac{1}{\Phi(h_2) - \Phi(h_1)} \int_{h_1}^{h_2}\left[\Phi\left(\frac{c_2 - \alpha z}{\sqrt{0.25}}\right) - \Phi\left(\frac{c_1 - \alpha z}{\sqrt{0.25}}\right)\right]\phi(z)\ud z, \nonumber \\
{\TCF}_{3}(c_2) &=& 1 - \frac{1}{1 - \Phi(h_2)} \int_{h_2}^{\infty}\Phi\left(\frac{c_2 - \alpha z}{\sqrt{0.25}}\right)\phi(z)\ud z, \nonumber
\end{eqnarray}
where $\phi(\cdot)$ denotes the density function of the standard normal random variable. We choose $\alpha = 0.5$.

The aim in this scenario is to compare  FI, MSI, IPW, SPE and KNN estimators when both the estimates for $\hat{\pi}_i$ and $\hat{\rho}_{ki}$ in the parametric approach are inconsistent. Therefore, $\hat{\rho}_{ki}$ could be obtained from a multinomial logistic regression model with $D = (D_1,D_2,D_3)$ as the response and $T$ as predictor. To estimate $\pi_i$, we use a generalized linear model for $V$ given $T$ and $A^{2/3}$ with logit link. Clearly, the two fitted models are misspecified. The KNN estimators are obtained by using $K = 1$ and $K=3$ and the Euclidean distance. Again, we use $\bar{K} = 2$ in the KNN procedure to estimate standard deviations of KNN estimators.

Table \ref{tab:res2} presents Monte Carlo means and standard deviations (across 5000 replications) for the estimators of the true class fractions, $\TCF_1$, $\TCF_2$ and $\TCF_3$.
The table also gives the means of the estimated standard deviations (of the estimators), based on the asymptotic theory.  The table clearly shows limitations of the (partially) parametric approaches in case of misspecified models for $\Pro(D_k = 1|T,A)$ and $\Pro(V = 1|T,A)$. More precisely, in term of bias, the FI, MSI, IPW and SPE approaches perform almost always poorly, with high distortion in almost all cases. As we mentioned in Section 2, the SPE estimators could fall outside the interval $(0,1)$. In our simulations, in the worst case, the estimator $\widehat{\TCF}_{3,\mathrm{SPE}}(-1.0, -0.5)$ gives rise to $20\%$ of the values  greater than $1$. Moreover, the Monte Carlo standard deviations shown in the table indicate that the SPE approach might yield unstable estimates. Finally, the misspecification also has a clear effect on the estimated standard deviations of the estimators.
On the other side, the estimators 1NN and 3NN seem to perform well in terms of both bias and standard deviation. In fact, KNN estimators yield estimated values that are near to the true values. In addition, we observe that the estimator 3NN has larger bias than 1NN, but with slightly less variance. 
%

\begin{table}[htbp]
\caption{Monte Carlo means, Monte Carlo standard deviations and estimated standard deviations of the estimators for true class fractions when both models for $\rho_k(t,a)$ and $\pi(t,a)$ are misspecified and sample size equals to $1000$. ``True'' denotes the true parameter value.}
\begin{center}
\begin{footnotesize}
\begin{tabular}{c c c c c c c c c c}
\toprule
	& TCF$_1$ & TCF$_2$ & TCF$_3$ & MC.sd$_1$ & MC.sd$_2$ & MC.sd$_3$ & asy.sd$_1$ & asy.sd$_2$ & asy.sd$_3$ \\
\midrule
	& \multicolumn{9}{c}{cut points $ = (-1.0,-0.5)$} \\
\hdashline
	True & 0.1812 & 0.1070 & 0.9817 &  &  &  &  &  &  \\
	FI & 0.1290 & 0.0588 & 0.9888 & 0.0153 & 0.0133 & 0.0118 & 0.0154 & 0.0087 & 0.0412 \\ 
 	MSI & 0.1299 & 0.0592 & 0.9895 & 0.0154 & 0.0153 & 0.0131 & 0.0157 & 0.0110 & 0.0417 \\ 
  	IPW & 0.1231 & 0.0576 & 0.9889 & 0.0178 & 0.0211 & 0.0208 & 0.0175 & 0.0207 & 0.3694 \\ 
  	SPE & 0.1407 & 0.0649 & 0.9877 & 0.0173 & 0.0216 & 0.0231 & 0.0176 & 0.0212 & 0.0432 \\ 
  	1NN & 0.1809 & 0.1036 & 0.9817 & 0.0224 & 0.0304 & 0.0255 & 0.0211 & 0.0262 & 0.0242 \\ 
  	3NN & 0.1795 & 0.0991 & 0.9814 & 0.0214 & 0.0258 & 0.0197 & 0.0208 & 0.0244 & 0.0240 \\
\midrule
	& \multicolumn{9}{c}{cut points $ = (-1.0,0.7)$} \\
\hdashline
	True & 0.1812 & 0.8609 & 0.4469 &  &  &  &  &  &  \\ 
	FI & 0.1290 & 0.7399 & 0.5850 & 0.0153 & 0.0447 & 0.1002 & 0.0154 & 0.0181 & 0.0739 \\ 
	MSI & 0.1299 & 0.7423 & 0.5841 & 0.0154 & 0.0453 & 0.1008 & 0.0157 & 0.0188 & 0.0666 \\ 
	IPW & 0.1231 & 0.7690 & 0.5004 & 0.0178 & 0.0902 & 0.2049 & 0.0175 & 0.0844 & 0.2018 \\ 
	SPE & 0.1407 & 0.7635 & 0.5350 & 0.0173 & 0.0702 & 0.2682 & 0.0176 & 0.0668 & 2.0344 \\ 
	1NN & 0.1809 & 0.8452 & 0.4406 & 0.0224 & 0.0622 & 0.1114 & 0.0211 & 0.0544 & 0.1079 \\ 
	3NN & 0.1795 & 0.8285 & 0.4339 & 0.0214 & 0.0521 & 0.0882 & 0.0208 & 0.0516 & 0.1066 \\
\midrule
	& \multicolumn{9}{c}{cut points $ = (-1.0,1.3)$} \\
\hdashline
	True & 0.1812 & 0.9732 & 0.1171 &  &  &  &  &  &  \\ 
	FI & 0.1290 & 0.9499 & 0.1900 & 0.0153 & 0.0179 & 0.0550 & 0.0154 & 0.0133 & 0.0422 \\ 
	MSI & 0.1299 & 0.9516 & 0.1902 & 0.0154 & 0.0184 & 0.0552 & 0.0157 & 0.0142 & 0.0389 \\ 
	IPW & 0.1231 & 0.9645 & 0.1294 & 0.0178 & 0.0519 & 0.1795 & 0.0175 & 0.0466 & 0.1344 \\ 
	SPE & 0.1407 & 0.9567 & 0.1760 & 0.0173 & 0.0425 & 0.3383 & 0.0176 & 0.0402 & 3.4770 \\ 
	1NN & 0.1809 & 0.9656 & 0.1124 & 0.0224 & 0.0218 & 0.0448 & 0.0211 & 0.0317 & 0.0710 \\ 
	3NN & 0.1795 & 0.9604 & 0.1086 & 0.0214 & 0.0172 & 0.0338 & 0.0208 & 0.0305 & 0.0716 \\
\midrule
	& \multicolumn{9}{c}{cut points $ = (-0.5,0.7)$} \\
\hdashline
	True & 0.4796 & 0.7539 & 0.4469 &  &  &  &  &  &  \\ 
	FI & 0.3715 & 0.6811 & 0.5850 & 0.0270 & 0.0400 & 0.1002 & 0.0151 & 0.0145 & 0.0739 \\ 
	MSI & 0.3723 & 0.6831 & 0.5841 & 0.0271 & 0.0409 & 0.1008 & 0.0162 & 0.0172 & 0.0666 \\ 
	IPW & 0.3547 & 0.7114 & 0.5004 & 0.0325 & 0.0883 & 0.2049 & 0.0322 & 0.0831 & 0.2018 \\ 
	SPE & 0.3949 & 0.6986 & 0.5350 & 0.0318 & 0.0687 & 0.2682 & 0.0331 & 0.0657 & 2.0344 \\ 
	1NN & 0.4783 & 0.7416 & 0.4406 & 0.0361 & 0.0610 & 0.1114 & 0.0311 & 0.0551 & 0.1079 \\ 
	3NN & 0.4756 & 0.7294 & 0.4339 & 0.0341 & 0.0499 & 0.0882 & 0.0304 & 0.0523 & 0.1066 \\
\midrule
	& \multicolumn{9}{c}{cut points $ = (-0.5,1.3)$} \\
\hdashline
	True & 0.4796 & 0.8661 & 0.1171 &  &  &  &  &  &  \\ 
 	FI & 0.3715 & 0.8910 & 0.1900 & 0.0270 & 0.0202 & 0.0550 & 0.0151 & 0.0142 & 0.0422 \\ 
	MSI & 0.3723 & 0.8924 & 0.1902 & 0.0271 & 0.0211 & 0.0552 & 0.0162 & 0.0165 & 0.0389 \\ 
	IPW & 0.3547 & 0.9068 & 0.1294 & 0.0325 & 0.0535 & 0.1795 & 0.0322 & 0.0492 & 0.1344 \\ 
	SPE & 0.3949 & 0.8918 & 0.1760 & 0.0318 & 0.0451 & 0.3383 & 0.0331 & 0.0435 & 3.4770 \\ 
	1NN & 0.4783 & 0.8620 & 0.1124 & 0.0361 & 0.0349 & 0.0448 & 0.0311 & 0.0390 & 0.0710 \\ 
	3NN & 0.4756 & 0.8613 & 0.1086 & 0.0341 & 0.0285 & 0.0338 & 0.0304 & 0.0371 & 0.0716 \\
\midrule  	
	& \multicolumn{9}{c}{cut points $ = (0.7,1.3)$} \\
\hdashline
	True & 0.9836 & 0.1122 & 0.1171 &  &  &  &  &  &  \\ 
	FI & 0.9618 & 0.2099 & 0.1900 & 0.0122 & 0.0317 & 0.0550 & 0.0043 & 0.0132 & 0.0422 \\ 
	MSI & 0.9613 & 0.2093 & 0.1902 & 0.0125 & 0.0320 & 0.0552 & 0.0048 & 0.0135 & 0.0389 \\ 
	IPW & 0.9548 & 0.1955 & 0.1294 & 0.0339 & 0.0831 & 0.1795 & 0.0323 & 0.0784 & 0.1344 \\ 
	SPE & 0.9582 & 0.1932 & 0.1760 & 0.0332 & 0.0618 & 0.3383 & 0.0320 & 0.0605 & 3.4770 \\ 
	1NN & 0.9821 & 0.1204 & 0.1124 & 0.0144 & 0.0494 & 0.0448 & 0.0133 & 0.0487 & 0.0710 \\ 
	3NN & 0.9804 & 0.1319 & 0.1086 & 0.0138 & 0.0404 & 0.0338 & 0.0131 & 0.0464 & 0.0716 \\  	
\bottomrule
\end{tabular}
\end{footnotesize}
\end{center}
\label{tab:res2}
\end{table}

\section{An illustration}

We use data on epithelial ovarian cancer (EOC) extracted from the Pre-PLCO Phase II Dataset from the SPORE/Early Detection Network/Prostate, Lung, Colon, and Ovarian Cancer Ovarian Validation Study. \footnote{The study protocol and data are publicly available at the address: \url{http://edrn.nci.nih.gov/protocols/119-spore-edrn -pre-plco-ovarian-phase-ii-val idation}.}

As in \cite{toduc:15}, we consider  the following three classes of EOC, i.e.,  benign disease, early stage (I and II) and late stage (III and IV) cancer, and 12 of the 59 available biomarkers, i.e.  CA125, CA153, CA72--4, Kallikrein 6 (KLK6), HE4, Chitinase (YKL40) and immune costimulatory protein--B7H4 (DD--0110), Insulin--like growth factor 2 (IGF2), Soluble mesothelin-related protein (SMRP), Spondin--2 (DD--P108), Decoy Receptor 3 (DcR3; DD--C248) and Macrophage inhibitory cytokine 1 (DD--X065). In addition, age of patients  is also considered.

After cleaning for missing data, we are left 134 patients with benign disease, 67 early stage samples and 77 late stage samples. 
As a preliminary step of our analysis we ranked the 12 markers according to 
value of VUS, estimated on the complete data. 
The observed ordering, consistent with medical knowledge, led us to select CA125 as
the test $T$ to be used to illustrate our method. 

To mimic verification bias, 
a  subset of the complete dataset is constructed using the test $T$ and a vector $A=(A_1, A_2)$ of two 
covariates, namely the marker CA153 ($A_1$) and age ($A_2$).
Reasons for using CA153 as a covariate come from the medical literature that suggests that 
the concomitant measurement of CA153 with CA125 could be advantageous in the pre-operative 
discrimination of benign and malignant ovarian tumors.  
In this subset, $T$ and $A$ are known for all samples (patients), but the true status (benign, early stage or late stage)  
is available  only for some samples, that we select according to the following mechanism. 
We select all samples having a value for $T$, $A_1$ and $A_2$ above their respective medians, i.e. 0.87,  0.30 and 45;
as for  the others, we apply the following selection process
\[
\Pro(V = 1) = 0.05 + 0.35\mathrm{I}(T > 0.87) + 0.25\mathrm{I}(A_1 > 0.30) + 0.35\mathrm{I}(A_2 > 45),
\]
leading to a marginal probability of selection equal to $0.634$. 

Since the test $T$ and the covariates $A_1, A_2$ are heterogeneous with respect to their variances,  the Mahalanobis distance is used for KNN estimators. Following discussion in Section 3.4, we use the selection rule (\ref{choice:K:1}) to find the size $K$ of the neighborhood. This leads to the choice of $K = 1$ for our data. In addition, we also employ $K = 3$ for the sake of comparison with 1NN result, and produce the estimate of the ROC surface based on full data (Full estimate), displayed in Figure \ref{fg:res:full}.
\begin{figure}[htbp]
\begin{center}
\includegraphics[width=0.5\linewidth]{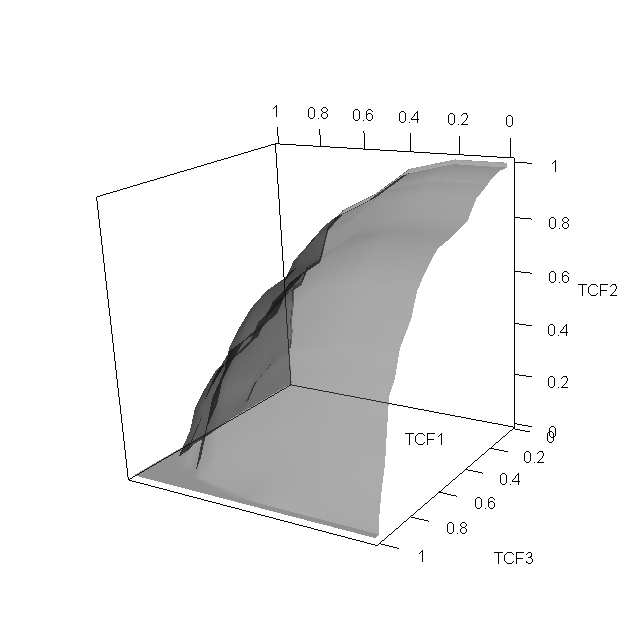}
\caption{Estimated ROC surface for the CA125 test, based on full data.}
\label{fg:res:full}
\end{center}
\end{figure}

Figure \ref{fg:res:knn} shows the 1NN and 3NN estimated ROC surfaces for the test $T$ (CA125). In this figure, we also give the 95\% ellipsoidal confidence regions (green color) for  $(\TCF_1,\TCF_2,\TCF_3)$ at cut points $(-0.56, 2.31)$. 
These regions are built using the asymptotic normality of the estimators.
Compared with the Full estimate,  KNN bias-corrected method proposed in the paper appears  well behave, yielding reasonable estimates of the ROC surface with incoplete data.

\begin{figure}[htbp]
\begin{center}
\begin{tabular}{@{}c@{}c@{}}
  \multicolumn{2}{c}{ }\\[-0.5cm]
  \includegraphics[width=0.5\linewidth]{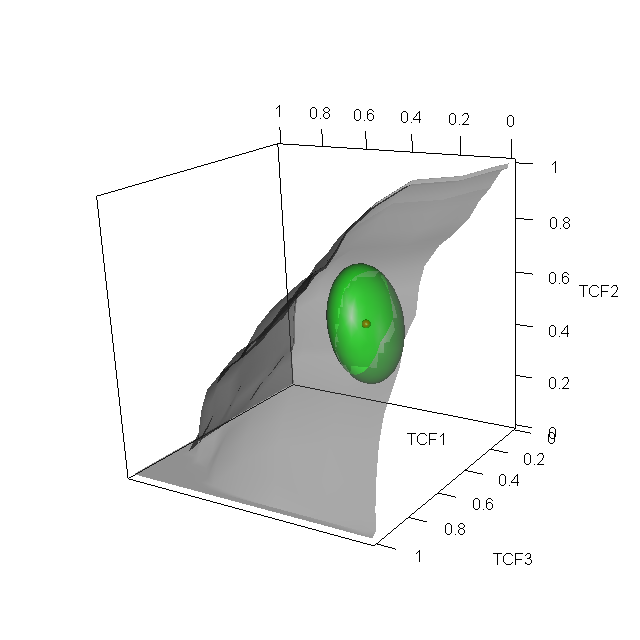} &
  \includegraphics[width=0.5\linewidth]{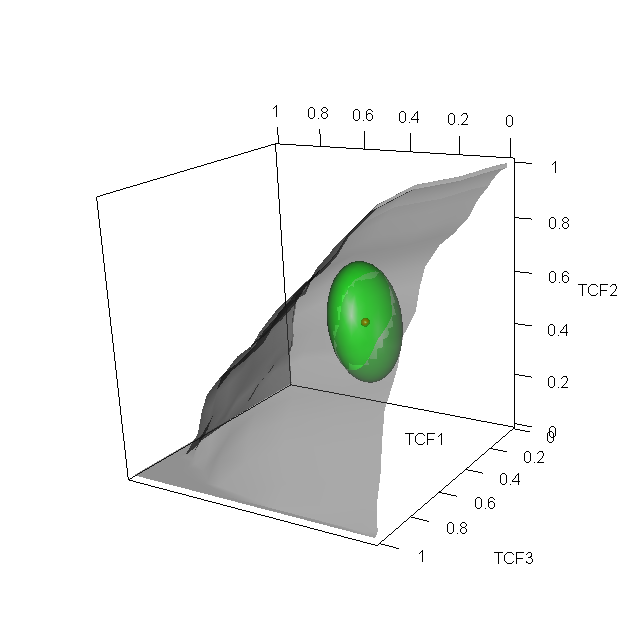} \\ [-0.1em]
      \footnotesize (a) 1NN & \footnotesize (b) 3NN \\
\end{tabular}
\caption{Bias--corrected estimated ROC surfaces for CA125 test, based on incomplete data.}
\label{fg:res:knn}
\end{center}
\end{figure}

%

\section{Conclusions}

A suitable solution for reducing the effects of model misspecification in statistical inference is to  resort to fully nonparametric methods. This paper proposes a nonparametric estimator of the ROC surface of a continuous-scale diagnostic test.  The estimator is based on nearest-neighbor imputation and works under MAR assumption.  It represents an alternative to (partially) parametric estimators discussed in \cite{toduc:15}. Our simulation results and the presented illustrative example show usefulness of the proposal. 

As in \cite{adi:15} and \cite{adi:15:auc},
a simple extension  of our estimator, that could be used when categorical auxiliary variables are also available, is possible.
Without loss of generality, we suppose that a single factor $C$, with $m$ levels,
is observed together with $T$ and $A$. We also assume that
$C$ may be associated with both $D$ and $V$. In this case, the sample can be divided into $m$ strata, i.e.
$m$ groups of units sharing the same level of $C.$  Then, for example, if  
the MAR assumption and first-order differentiability of
the functions $\rho_k(t,a)$ and $\pi(t,a)$ hold in each stratum,
a consistent and asymptotically normally distributed estimator of TCF$_1$  is 
\begin{equation}\nonumber
\widehat{\TCF}^S_{1,\mathrm{KNN}}(c_1) =\frac{1}{n}\sum_{j=1}^m  n_j
\widehat{\TCF}_{1j,\mathrm{KNN}}^{cond}(c_1),
\end{equation}
where $n_j$ denotes the size of the $j$-th  stratum and $\widehat{\TCF}^{cond}_{1j,\mathrm{KNN}}(c_1) $ denotes the KNN estimator of the conditional TCF$_1$, i.e., the KNN estimator in (\ref{est:knn3}) obtained from the patients  in the $j$-th stratum. 
Of course, we must assume that, for every $j$, ratios $n_j/n$ have finite and nonzero limits as 
$n$ goes to infinity.


\appendix

\section{Appendix 1} 
According the proof of Theorem \ref{thr:knn:2}, we have
\begin{equation}
\left(\hat{\beta}_{12,\mathrm{KNN}} - \hat{\beta}_{22,\mathrm{KNN}}\right) - \left(\beta_{12} - \beta_{22}\right) \simeq \left(S_{12} - S_{22}\right) + \left(R_{12} - R_{22}\right) + \left(W_{12} - W_{22}\right) + o_{p}(n^{-1/2}).
\label{eq:rem:1}
\end{equation}
Here, we have
\begin{eqnarray}
S_{12} - S_{22} &=& \frac{1}{n}\sum_{i=1}^{n} V_i \mathrm{I}(c_1 \le T_i < c_2)\left(D_{2i} - \rho_{2i}\right) \nonumber \\
R_{12} - R_{22} &=& \frac{1}{n}\sum_{i=1}^{n}\left[\mathrm{I}(c_1 \le T_i < c_2)\rho_{2i} - (\beta_{12} - \beta_{22}) \right] \nonumber \\
W_{12} - W_{22} &=& \frac{1}{n}\sum_{i=1}^{n}\mathrm{I}(c_1 \le T_i < c_2)(1-V_i)\left[\frac{1}{K}\sum_{l=1}^{K}\left(V_{i(l)}D_{2i(l)} - \rho_{2i(l)}\right)\right] \nonumber.
\end{eqnarray}
Under that, we realize that quantities $S_{12} - S_{22}$ and $S_{jk}$, so as $R_{12} - R_{22}$ and $R_{jk}$, and $W_{12} - W_{22}$ and $W_{jk}$, play, in essence, a similar role. Therefore, the quantities in right hand side of equation (\ref{eq:rem:1}) have approximately normal distributions with mean $0$ and variances
\begin{eqnarray}
\V \left(\sqrt{n}(S_{12} - S_{22})\right) &=& \E \left\{\pi(T,A)\delta^2(T,A) \right\}, \nonumber \\
\V \left(\sqrt{n} (R_{12} - R_{22}) \right) &=& \V \left[\mathrm{I}(c_1 \le T_i < c_2)\rho_{2}(T,A)\right], \nonumber \\
\V \left(\sqrt{n} (W_{12} - W_{22}) \right) &=&  \frac{1}{K}\E \left[(1 - \pi(T,A))\delta^2(T,A)\right] + \E \left[\frac{(1-\pi(T,A))^2\delta^2(T,A)}{\pi(T,A)}\right] .\nonumber
\end{eqnarray}
where, $\delta^2(T,A)$ is the conditional variance of $\mathrm{I}(c_1 \le T_i < c_2, D_{2i} = 1)$ given $T,A$. Then, we get
\[
\sqrt{n}\left[\left(\hat{\beta}_{12,\mathrm{KNN}} - \hat{\beta}_{22,\mathrm{KNN}}\right) - \left(\beta_{12} - \beta_{22}\right)\right] \stackrel{d}{\to} \mathcal{N}(0,\lambda^2).
\]
To obtain $\lambda^2$, we notice that the quantities $R_{12} - R_{22}$ and $(S_{12} - S_{22}) + (W_{12} - W_{22})$ are uncorrelated and the asymptotic covariance of $S_{12} - S_{22}$ and $W_{12} - W_{22}$ equals to $\E \left[(1 - \pi(T,A))\delta^2(T,A)\right]$. Taking the sum of this covariance and the above variances, the desired asymptotic variance $\lambda^2$ is approximately 
\begin{equation}
\lambda^2 = \left\{(\beta_{12} - \beta_{22})\left[1 - (\beta_{12} - \beta_{22})\right] + \omega^2_{12} - \omega^2_{22}\right\}.
\label{eq:rem:2}
\end{equation}

\section{Appendix 2}

Here, we focus on the elements $\xi_{12}$, $\xi_{13}$ and $\xi_{23}$ of the covariance mtrix $\Xi$.
We can write
\begin{eqnarray}
\xi_{12} &=& - \frac{1}{\theta_1 \theta_2} \left(\sigma_{1112} - \sigma_{1122}\right) + \frac{\beta_{11}}{\theta_1^2 \theta_2}\left(\sigma_{112} - \sigma_{122}\right) - \frac{\beta_{12} - \beta_{22}}{\theta_2^2} \left(\frac{\beta_{11}}{\theta_1^2} \sigma_{12}^* - \frac{\sigma_{211}}{\theta_1}\right), 
\label{cov:12} \\
\xi_{13} &=& \frac{1}{1 - \theta_1 - \theta_2}\left(\frac{\beta_{11}}{\theta_1^2}\sigma_{123} - \frac{\sigma_{1123}}{\theta_1}\right)  \nonumber \\
&& + \: \frac{\beta_{23}}{\theta_1(1 - \theta_1 - \theta_2)^2}\left[
\frac{\beta_{11}}{\theta_1} \left(\sigma_1^2 + \sigma_{12}^*\right) - \left(\sigma_{111} + \sigma_{211}\right)
\right], 
\label{cov:13}
\end{eqnarray}
and
\begin{eqnarray}
\xi_{23} &=& \frac{1}{\theta_2(1 - \theta_1 - \theta_2)}\left[\left(\sigma_{1223} - \sigma_{2223}\right) - \frac{\beta_{12} - \beta_{22}}{\theta_2}\sigma_{223}\right] \nonumber \\
&& + \: \frac{\beta_{23}}{\theta_2(1 - \theta_1 - \theta_2)^2}\left[\left(\sigma_{112} - \sigma_{122} + \sigma_{212} - \sigma_{222}\right) - \frac{\beta_{12} - \beta_{22}}{\theta_2}\left(\sigma_2^2 + \sigma_{12}^*\right)\right].
\label{cov:23}
\end{eqnarray}
Recall that
\begin{eqnarray}
\hat{\theta}_{k,\mathrm{KNN}} - \theta_{k} &=& \frac{1}{n}\sum_{i=1}^{n}\left[V_iD_{ki} + (1-V_i)\rho_{ki}\right] + \frac{1}{n}\sum_{i=1}^{n} (1-V_i)(\hat{\rho}_{ki,K} - \rho_{ki}) - \theta_{k} \nonumber\\
&=& \frac{1}{n}\sum_{i=1}^{n}V_i\left[D_{ki} - \rho_{ki}\right] + \frac{1}{n}\sum_{i=1}^{n}\left[\rho_{ki} - \theta_{k}\right]\nonumber \\
&& + \: \frac{1}{n}\sum_{i=1}^{n}\left[\frac{1}{K}\sum_{l=1}^{K}\left(V_{i(l)}D_{ki(l)} - \rho_{ki(l)}\right)\right] + o_p\left(n^{-1/2}\right) \nonumber \\
&=& S_{k} + R_{k} + W_{k} + o_p\left(n^{-1/2}\right); \nonumber
\end{eqnarray}
and
\begin{eqnarray}
\hat{\beta}_{jk,\mathrm{KNN}} - \beta_{jk} &=& \frac{1}{n}\sum_{i=1}^{n}\mathrm{I}(T_i \ge c_j)\left[V_iD_{ki} + (1-V_i)\rho_{ki}\right] + \frac{1}{n}\sum_{i=1}^{n}\mathrm{I}(T_i \ge c_j) (1-V_i)(\hat{\rho}_{ki,K} - \rho_{ki}) - \beta_{jk} \nonumber\\
&=& \frac{1}{n}\sum_{i=1}^{n}\mathrm{I}(T_i \ge c_j)V_i\left[D_{ki} - \rho_{ki}\right] + \frac{1}{n}\sum_{i=1}^{n}\left[\mathrm{I}(T_i \ge c_j)\rho_{ki} - \beta_{jk}\right]\nonumber \\
&& + \: \frac{1}{n}\sum_{i=1}^{n}\mathrm{I}(T_i \ge c_j)(1-V_i)\left[\frac{1}{K}\sum_{l=1}^{K}\left(V_{i(l)}D_{ki(l)} - \rho_{ki(l)}\right)\right] + o_p\left(n^{-1/2}\right) \nonumber \\
&=& S_{jk} + R_{jk} + W_{jk} + o_p\left(n^{-1/2}\right). \nonumber
\end{eqnarray}
Then, we  restating some terms that appear in expressions (\ref{cov:12})--(\ref{cov:23}).
First, we consider the term, $\sigma_{1112} - \sigma_{1122}$. We have
\begin{eqnarray}
\sigma_{1112} - \sigma_{1122} &=& \aCov(\sqrt{n}\hat{\beta}_{11,\mathrm{KNN}},\sqrt{n}\hat{\beta}_{12,\mathrm{KNN}}) - \aCov(\sqrt{n}\hat{\beta}_{11,\mathrm{KNN}},\sqrt{n}\hat{\beta}_{22,\mathrm{KNN}}) \nonumber \\
&=& \aCov(\sqrt{n}\hat{\beta}_{11,\mathrm{KNN}},\sqrt{n}\hat{\beta}_{12,\mathrm{KNN}} - \sqrt{n}\hat{\beta}_{22,\mathrm{KNN}}) \nonumber \\
&=& \aCov\left(\sqrt{n}(S_{11} + R_{11} + W_{11}),\sqrt{n}(S_{12} - S_{22}) + \sqrt{n}(R_{12} - R_{22}) + \sqrt{n}(W_{12} - W_{22})\right) \nonumber \\
&=& \aCov\left(\sqrt{n}S_{11},\sqrt{n}(S_{12} - S_{22})\right) + \aCov\left(\sqrt{n}S_{11},\sqrt{n}(W_{12} - W_{22})\right) \nonumber\\
&& + \: \aCov\left(\sqrt{n}R_{11},\sqrt{n}(R_{12} - R_{22})\right) + \aCov\left(\sqrt{n}W_{11},\sqrt{n}(S_{12} - S_{22})\right) \nonumber \\
&& + \: \aCov\left(\sqrt{n}W_{11},\sqrt{n}(W_{12} - W_{22})\right). \nonumber
\end{eqnarray}
This result follows from the fact that
 $\sqrt{n}R_{11}$ and $\sqrt{n}(S_{12} - S_{22}) + \sqrt{n}(W_{12} - W_{22})$,  and $\sqrt{n}(S_{11} + W_{11})$ and $\sqrt{n}(R_{12} - R_{22})$ are uncorrelated (see also \cite{cheng:94}). By arguments similar to those used in
 \cite{ning:12}, we also obtain
\begin{eqnarray}
\aCov\left(\sqrt{n}S_{11},\sqrt{n}(S_{12} - S_{22})\right) &=& \E \left\{\pi(T,A) \Cov(\mathrm{I}(T \ge c_1)D_1,\mathrm{I}(c_1 \le T < c_2)D_2|T,A )\right\} \nonumber \\
&=& \E \left\{\pi(T,A) \mathrm{I}(c_1 \le T < c_2)\Cov(D_1,D_2|T,A )\right\} \nonumber \\
&=& -\E \left\{\pi(T,A) \mathrm{I}(c_1 \le T < c_2)\rho_{1}(T,A)\rho_{2}(T,A)\right\} \nonumber.
\end{eqnarray}
Similarly, we have that
\begin{eqnarray}
\aCov\left(\sqrt{n}S_{11},\sqrt{n}(W_{12} - W_{22})\right) &=& -\E\left\{[1-\pi(T,A)]\mathrm{I}(c_1 \le T < c_2)\rho_{1}(T,A)\rho_{2}(T,A)\right\}, \nonumber \\
\aCov\left(\sqrt{n}R_{11},\sqrt{n}(R_{12} - R_{22})\right) &=&  - \beta_{11}(\beta_{12} - \beta_{22}) + \E\left\{\mathrm{I}(c_1 \le T < c_2)\rho_{1}(T,A)\rho_{2}(T,A)\right\}, \nonumber \\
\aCov\left(\sqrt{n}W_{11},\sqrt{n}(S_{12} - S_{22})\right) &=& -\E\left\{[1-\pi(T,A)]\mathrm{I}(c_1 \le T < c_2)\rho_{1}(T,A)\rho_{2}(T,A)\right\}, \nonumber \\
\aCov\left(\sqrt{n}W_{11},\sqrt{n}(W_{12} - W_{22})\right) &=& -\frac{1}{K}\E\left\{[1-\pi(T,A)]\mathrm{I}(c_1 \le T < c_2)\rho_{1}(T,A)\rho_{2}(T,A)\right\} \nonumber\\
&& - \E\left\{\frac{[1-\pi(T,A)]^2\mathrm{I}(c_1 \le T < c_2)\rho_{1}(T,A)\rho_{2}(T,A)}{\pi(T,A)}\right\} . \nonumber
\end{eqnarray}
This leads to 
\begin{eqnarray}
\sigma_{1112} - \sigma_{1122} &=& -\Bigg[  
\psi^2_{1212} +  \beta_{11}(\beta_{12} - \beta_{22}) \Bigg],
\label{sig_1112 - sig_1122}
\end{eqnarray}
where
\begin{eqnarray}
\psi^2_{1212} &=&
\left(1 + \frac{1}{K}\right) \E \left\{[1-\pi(T,A)]\mathrm{I}(c_1 \le T < c_2)\rho_{1}(T,A)\rho_{2}(T,A) \right\} \nonumber \\
&& + \: \E\left\{\frac{[1-\pi(T,A)]^2\mathrm{I}(c_1 \le T < c_2)\rho_{1}(T,A)\rho_{2}(T,A)}{\pi(T,A)}\right\}. \nonumber 
\end{eqnarray}

Second, we consider $\sigma_{112} - \sigma_{122}$. In this case, we have
\begin{eqnarray}
\sigma_{112} - \sigma_{122} &=& \aCov(\sqrt{n}\hat{\theta}_{1,\mathrm{KNN}},\sqrt{n}\hat{\beta}_{12,\mathrm{KNN}}) - \aCov(\sqrt{n}\hat{\theta}_{1,\mathrm{KNN}},\sqrt{n}\hat{\beta}_{22,\mathrm{KNN}}) \nonumber \\
&=& \aCov\left(\sqrt{n}\hat{\theta}_{1,\mathrm{KNN}},\sqrt{n}(\hat{\beta}_{12,\mathrm{KNN}} - \hat{\beta}_{22,\mathrm{KNN}})\right) \nonumber \\
&=& \aCov\left(\sqrt{n}(S_{1} + R_{1} + W_{1}),\sqrt{n}(S_{12} - S_{22}) + \sqrt{n}(R_{12} - R_{22}) + \sqrt{n}(W_{12} - W_{22})\right) \nonumber \\
&=& \aCov\left(\sqrt{n}S_{1},\sqrt{n}(S_{12} - S_{22})\right) + \aCov\left(\sqrt{n}S_{1},\sqrt{n}(W_{12} - W_{22})\right) \nonumber\\
&& + \: \aCov\left(\sqrt{n}R_{1},\sqrt{n}(R_{12} - R_{22})\right) + \aCov\left(\sqrt{n}W_{1},\sqrt{n}(S_{12} - S_{22})\right) \nonumber \\
&& + \: \aCov\left(\sqrt{n}W_{1},\sqrt{n}(W_{12} - W_{22})\right). \nonumber
\end{eqnarray}
We obtain
\begin{eqnarray}
\aCov\left(\sqrt{n}S_{1},\sqrt{n}(S_{12} - S_{22})\right) &=& -\E \left\{\pi(T,A) \mathrm{I}(c_1 \le T < c_2)\rho_{1}(T,A)\rho_{2}(T,A)\right\}, \nonumber \\
\aCov\left(\sqrt{n}S_{1},\sqrt{n}(W_{12} - W_{22})\right) &=& -\E\left\{[1-\pi(T,A)]\mathrm{I}(c_1 \le T < c_2)\rho_{1}(T,A)\rho_{2}(T,A)\right\}, \nonumber \\
\aCov\left(\sqrt{n}R_{1},\sqrt{n}(R_{12} - R_{22})\right) &=&  - \theta_{1}(\beta_{12} - \beta_{22}) + \E\left\{\mathrm{I}(c_1 \le T < c_2)\rho_{1}(T,A)\rho_{2}(T,A)\right\}, \nonumber \\
\aCov\left(\sqrt{n}W_{1},\sqrt{n}(S_{12} - S_{22})\right) &=& -\E\left\{[1-\pi(T,A)]\mathrm{I}(c_1 \le T < c_2)\rho_{1}(T,A)\rho_{2}(T,A)\right\}, \nonumber \\
\aCov\left(\sqrt{n}W_{1},\sqrt{n}(W_{12} - W_{22})\right) &=& -\frac{1}{K}\E\left\{[1-\pi(T,A)]\mathrm{I}(c_1 \le T < c_2)\rho_{1}(T,A)\rho_{2}(T,A)\right\} \nonumber\\
&& - \E\left\{\frac{[1-\pi(T,A)]^2\mathrm{I}(c_1 \le T < c_2)\rho_{1}(T,A)\rho_{2}(T,A)}{\pi(T,A)}\right\} , \nonumber
\end{eqnarray}
and then
\begin{eqnarray}
\sigma_{112} - \sigma_{122} &=& - [
\psi^2_{1212} +  \theta_{1}(\beta_{12} - \beta_{22}) ].
\label{sig_112 - sig_122}
\end{eqnarray}\\

Similarly, it is straightforward to obtain 
\begin{eqnarray}
\sigma_{211} &=& -[  
\psi^2_{112}+  \theta_{2}\beta_{11} ]
\label{sig_211}
\end{eqnarray}
and
\begin{eqnarray}
\sigma_{123} &=& - [  
\psi^2_{213}+  \theta_{1}\beta_{23}],
\label{sig_123}
\end{eqnarray}
with
\begin{eqnarray}
\psi^2_{112}&=&
\left(1 + \frac{1}{K}\right) \E \left\{[1-\pi(T,A)]\mathrm{I}(T \ge c_1)\rho_{1}(T,A)\rho_{2}(T,A) \right\} \nonumber \\
&& + \: \E\left\{\frac{[1-\pi(T,A)]^2\mathrm{I}(T \ge c_1)\rho_{1}(T,A)\rho_{2}(T,A)}{\pi(T,A)}\right\} \nonumber 
\end{eqnarray}
and
\begin{eqnarray}
\psi^2_{213}&=&
\left(1 + \frac{1}{K}\right) \E \left\{[1-\pi(T,A)]\mathrm{I}(T \ge c_2)\rho_{1}(T,A)\rho_{3}(T,A) \right\} \nonumber \\
&& + \: \E\left\{\frac{[1-\pi(T,A)]^2\mathrm{I}(T \ge c_2)\rho_{1}(T,A)\rho_{3}(T,A)}{\pi(T,A)}\right\}. \nonumber 
\end{eqnarray}

The covariance between $\sqrt{n}\hat{\theta}_{1,\mathrm{KNN}}$ and $\sqrt{n}\hat{\theta}_{2,\mathrm{KNN}}$ is computed analogously, i.e.,
\begin{eqnarray}
\sigma_{12}^* &=& -[ \theta_{1}\theta_{2} +   \psi^2_{12}
],
\label{sig_12*}
\end{eqnarray}
where
\begin{eqnarray}
\psi^2_{12}&=&
\left(1 + \frac{1}{K}\right) \E \left\{[1-\pi(T,A)]\rho_{1}(T,A)\rho_{2}(T,A) \right\} \nonumber \\
&& + \: \E\left\{\frac{[1-\pi(T,A)]^2\rho_{1}(T,A)\rho_{2}(T,A)}{\pi(T,A)}\right\}. \nonumber 
\end{eqnarray}

By using results (\ref{sig_1112 - sig_1122}), (\ref{sig_112 - sig_122}), (\ref{sig_211}) and (\ref{sig_12*}) into (\ref{cov:12}), we can obtain a suitable expression for  $\aCov\left(\sqrt{n}\widehat{\TCF}_{1,\mathrm{KNN}} (c_1),\sqrt{n}\widehat{\TCF}_{2,\mathrm{KNN}} (c_1,c_2)\right)$, which depends on easily estimable quanties.

Clearly, a similar approach can be used to get suitable expressions for $\xi_{13}$ and $\xi_{23}$ too.
In particular, the estimable version of $\xi_{13}$ can be obtained by using suitable expressions for $\sigma_{123}$, $\sigma_{1123}$
and $\sigma_{111} + \sigma_{211}$. 
The quantity $\sigma_{123}$ is already computed in (\ref{sig_123}), and the formula for $\sigma_{1123}$ can be obtained as
\begin{eqnarray}
\sigma_{1123} &=& -[  
\psi^2_{213} +  \beta_{11}\beta_{23}].
\nonumber \label{sig_1123}
\end{eqnarray}
To compute $\sigma_{111} + \sigma_{211}$, we notice that
\begin{eqnarray}
\aCov\left(\sqrt{n}\hat{\theta}_{3,\mathrm{KNN}},\sqrt{n}\hat{\beta}_{11,\mathrm{KNN}}\right) &=& \aCov\left(\sqrt{n} - \sqrt{n}(\hat{\theta}_{1,\mathrm{KNN}} + \hat{\theta}_{1,\mathrm{KNN}}),\sqrt{n}\hat{\beta}_{11,\mathrm{KNN}}\right)\nonumber \\
&=& - \aCov\left(\sqrt{n}(\hat{\theta}_{1,\mathrm{KNN}} + \hat{\theta}_{1,\mathrm{KNN}}),\sqrt{n}\hat{\beta}_{11,\mathrm{KNN}}\right)\nonumber.
\end{eqnarray}
It leads to $\sigma_{111} + \sigma_{211} = -\sigma_{311}$. Similarly to (\ref{sig_211}), we have that
\begin{eqnarray}
\sigma_{311} &=& - [  
\psi^2_{113} +  \theta_{3}\beta_{11}],
\nonumber \label{sig_311}
\end{eqnarray}
where
\begin{eqnarray}
\psi^2_{113}&=&
\left(1 + \frac{1}{K}\right) \E \left\{[1-\pi(T,A)]\mathrm{I}(T \ge c_1)\rho_{1}(T,A)\rho_{3}(T,A) \right\} \nonumber \\
&& + \: \E\left\{\frac{[1-\pi(T,A)]^2\mathrm{I}(T \ge c_1)\rho_{1}(T,A)\rho_{3}(T,A)}{\pi(T,A)}\right\}. \nonumber 
\end{eqnarray}

For the last term $\xi_{23}$, we need to make some other calculations. First, the quantity $\sigma_{1223} - \sigma_{2223}$ is obtained as $\sigma_{1112} - \sigma_{1122}$.
We have 
\[
\sigma_{1223} - \sigma_{2223}=-\beta_{23}(\beta_{12} - \beta_{22}),
\] 
because $\mathrm{I}(c_1 \le T < c_2) \mathrm{I}(T \ge c_2) = 0$. Second, the term $\sigma_{223}$ is obtained as
\begin{eqnarray}
\sigma_{223} &=& - [  
\psi^2_{223}+  \theta_{2}\beta_{23}],
\nonumber \label{sig_223}
\end{eqnarray}
where
\begin{eqnarray}
\psi^2_{223}&=&\left(1 + \frac{1}{K}\right) \E \left\{[1-\pi(T,A)]\mathrm{I}(T \ge c_2)\rho_{2}(T,A)\rho_{3}(T,A) \right\} \nonumber \\
&& + \: \E\left\{\frac{[1-\pi(T,A)]^2\mathrm{I}(T \ge c_2)\rho_{2}(T,A)\rho_{3}(T,A)}{\pi(T,A)}\right\} . \nonumber 
\end{eqnarray}
Moreover, it is straightforward to show that
\[
-(\sigma_{312} - \sigma_{322}) = \sigma_{112} - \sigma_{122} + \sigma_{212} - \sigma_{222},
\]
and that
\begin{eqnarray}
\sigma_{312} - \sigma_{322} &=& -[  
\psi^2_{1223} +  \theta_{3}(\beta_{12} - \beta_{22}) ],
\nonumber \label{sig_312 - sig_322}
\end{eqnarray}
with
\begin{eqnarray}
\psi^2_{1223}&=&\left(1 + \frac{1}{K}\right) \E \left\{[1-\pi(T,A)]\mathrm{I}(c_1 \le T < c_2)\rho_{2}(T,A)\rho_{3}(T,A) \right\} \nonumber \\
&& + \: \E\left\{\frac{[1-\pi(T,A)]^2\mathrm{I}(c_1 \le T < c_2)\rho_{2}(T,A)\rho_{3}(T,A)}{\pi(T,A)}\right\}. \nonumber 
\end{eqnarray}
%

\end{document}